\documentclass[journal]{IEEEtran}
\usepackage{graphicx}
\usepackage{epstopdf}
\usepackage{algorithm}
\usepackage{algpseudocode}
\usepackage{pifont}
\usepackage{epsfig,amsfonts}
\usepackage{caption}
\usepackage{subcaption}
\usepackage{amsmath}
\usepackage{amsthm}
\usepackage{bbm}
\usepackage{amssymb}
\usepackage{url}
\usepackage{enumitem}
\usepackage{xcolor}
\usepackage{multirow}

\theoremstyle{definition}
\newtheorem{definition}{Definition}
\newtheorem{proposition}{Proposition}
\newtheorem{theorem}{Theorem}

\newtheorem{lemma}{Lemma}

{
	\theoremstyle{plain}
	
}

\newtheorem{innercustomgeneric}{\customgenericname}
\providecommand{\customgenericname}{}
\newcommand{\newcustomtheorem}[2]{%
	\newenvironment{#1}[1]
	{%
		\renewcommand\customgenericname{#2}%
		\renewcommand\theinnercustomgeneric{##1}%
		\innercustomgeneric
	}
	{\endinnercustomgeneric}
}

\newcustomtheorem{customthm}{Theorem}
\newcustomtheorem{customlemma}{Lemma}
\newcustomtheorem{customDef}{Definition}

\usepackage{stackengine}

\ifCLASSINFOpdf

\else

\fi

\begin{document}

\title{Frequency-Resolved Optical Gating Recovery via Smoothing Gradient }

\author{Samuel Pinilla,~\IEEEmembership{Student Member,~IEEE,}
        Tamir Bendory,
        Yonina C. Eldar,~\IEEEmembership{Fellow Member,~IEEE,}
        and~Henry Arguello,~\IEEEmembership{Senior Member,~IEEE}
        \thanks{S. Pinilla is with the Department of Electrical Engineering, Universidad Industrial de Santander, Bucaramanga, Santander, 680002 Colombia. e-mail: samuel.pinilla@correo.uis.edu.co.}
        \thanks{T. Bendory is with the Program in Applied and Computational Mathematics, Princeton University, Princeton, NJ, USA. e-mail: tamir.bendory@princeton.edu.}
        \thanks{Y. C. Eldar is with the Weizmann Institute of Science, Rehovot, 7610001 Israel. e-mail: yonina.eldar@weizmann.ac.il.}
        \thanks{H. Arguello is with the Department of Computer Science, Universidad Industrial de Santander, Bucaramanga, Santander, 680002 Colombia. e-mail: henarfu@uis.edu.co.}
        \thanks{Manuscript received November XX, 201X; revised August XX, 201X.}\vspace{-1.5em}}

\markboth{Journal of \LaTeX\ Class Files,~Vol.~14, No.~8, August~2015}%
{Shell \MakeLowercase{\textit{et al.}}: Bare Demo of IEEEtran.cls for IEEE Journals}

\maketitle

\begin{abstract}
Frequency-resolved optical gating (FROG) is a popular technique for complete characterization of ultrashort laser pulses. The acquired data in FROG, called FROG trace, is the Fourier magnitude of the product of the unknown pulse with a time-shifted version of itself, for several different shifts. To estimate the pulse from the FROG trace, we propose an algorithm that minimizes a smoothed non-convex least-squares objective function. The method consists of two steps. First, we approximate the pulse by an iterative spectral algorithm. Then, the attained initialization is refined based upon a sequence of block stochastic gradient iterations. The algorithm is theoretically simple, numerically scalable, and easy-to-implement. Empirically, our approach outperforms the state-of-the-art when the FROG trace is incomplete, that is, when only few shifts are recorded. Simulations also suggest that the proposed algorithm exhibits similar computational cost compared to a state-of-the-art technique for both complete and incomplete data. In addition, we prove that in the vicinity of the true solution, the algorithm converges to a critical point. A Matlab implementation is publicly available at \url{https://github.com/samuelpinilla/FROG}.
\end{abstract}


\IEEEpeerreviewmaketitle

\section{Introduction}

\IEEEPARstart{P}{hase} retrieval (PR) is the inverse problem of recovering a signal from its Fourier magnitude \cite{bendory2017fourier}. PR arises in many fields in science and engineering, such as optics \cite{xu2015overcoming,shechtman2015phase}, astronomical imaging \cite{fienup1987phase}, microscopy \cite{mayo2003x}, and X-ray crystallography \cite{millane1990phase}. In this work, we focus on a popular technique for full characterization of ultrashort pulses called frequency-resolved optical gating (FROG) \cite{trebino2012frequency,trebino1997measuring}. The acquired data in FROG corresponds to the squared Fourier magnitude of the product of the unknown pulse with its delayed replica, for several different time shifts. The product of the signal with itself is usually performed using a second harmonic generation crystal. This measured data is called the \textit{FROG trace}. In this paper, we focus on the inverse problem of recovering a pulse from its second-harmonic generation FROG trace.

Recent works have studied conditions under which a pulse can be uniquely identified, up to trivial ambiguities, from its FROG trace \cite{7891012,bendory2017signal}. In Section \ref{sec:background} we present and discuss these results. In particular, it has been shown that in theory not all the delay steps are needed to recover the pulse. The most commonly used algorithm to estimate a pulse from its FROG trace is the principal component generalized projections (PCGP), originally introduced in \cite{kane1999recent}. PCGP follows classical algorithms in PR based on alternating projections. Specifically, PCGP is initialized by a Gaussian pulse with random phases. It then builds an auxiliary matrix by rearranging the columns in the time-delay plane which is Fourier transformed. The pulse is next updated as the leading eigenvector of this matrix. The FROG trace of the updated pulse is then constructed and its magnitude is replaced by the acquired phaseless measurements. This procedure is repeated until convergence.

A recent algorithm, called Retrieved-Amplitude N-grid Algorithmic (RANA) \cite{jafari2019100}, exploits the expected continuity of the signal to construct several initial estimates of the power spectrum of the pulse. These estimates are obtained from a set of smaller grids of the FROG trace, a strategy called multi-grid. The initial guesses are then refined using the generalized projections method in \cite{delong1994pulse}, in parallel, where the one that best fits the acquired measurements is chosen as the reconstructed pulse. The multi-grid initialization procedure allows RANA to converge faster than PCGP and leads to more accurate reconstructions. Importantly, both RANA and PCGP were developed to retrieve the pulse when all the delay steps of the FROG trace are acquired.

In \cite{sidorenko2016ptychographic}, the authors suggest an alternative recovery strategy, inspired by blind ptychography \cite{bendory2017fourier}. This method starts with the integrated measured FROG trace over the frequency dimension as an initial guess. Then, the initialization is refined using a stochastic descent strategy, which involves a single time shift of the FROG trace per iteration to update the estimated pulse. This paper claims to attain better estimation of the pulse compared to PCGP. However, the non-smooth cost function may increase the amount of required measurements (sample complexity) to recover the pulse, which affects the success rate as will be shown in the numerical results. This limitation appears since the non-smoothness of the objective may lead to unbounded gradients \cite{8410803,wang2018phase}.

In this paper we propose a block stochastic gradient algorithm (BSGA) for FROG recovery that minimizes a smoothed amplitude-based least-squares empirical loss. Amplitude-based objectives have shown improved results in standard PR \cite{pauwels2018fienup}. BSGA is initialized by a spectral method that requires only few iterations. This procedure can be seen as a modification of the strategy proposed in \cite{bendory2018non} that approximates the signal $\mathbf{x}$ from the FROG trace as the leading eigenvector of a carefully designed matrix. The two stages (initialization and gradient iterations) differ from prior contributions by the new initialization technique and the inclusion of a smoothing function, following \cite{8410803}. Specifically, in contrast to \cite{jafari2019100}, the proposed initialization aims to estimate the pulse, rather than its power spectrum. Numerical results show that our initialization returns a more accurate estimation of the pulse compared to the starting point of the ptychography (Ptych) method in \cite{sidorenko2016ptychographic}. Furthermore, BSGA shows improvements in recovering both the magnitude and the phase of the pulse from incomplete data ($L>1$) compared to Ptych. As aforementioned, alternative methods such as PCGP and RANA are designed for $L=1$, and therefore do not work well in this regime. In Theorem \ref{theo:contraction} we provide partial theoretical justification for the success of the algorithm by showing that in the vicinity of the true solution, BSGA converges to a critical point. 

The paper is organized as follows. We begin in Section \ref{sec:background} by introducing necessary background on FROG. Section \ref{sec:algorithm} presents the proposed initialization technique and introduces an iterative procedure to refine the solution by minimizing a smooth least-squares objective. Section \ref{sec:numerical} presents numerical results and compares our approach with competitive algorithms. Finally, Section \ref{sec:conclusion} concludes the paper.

We denote by $\mathbb{R}_{+}:=\{w\in \mathbb{R}: w\geq 0\}$ and $\mathbb{R}_{++}:=\{w\in \mathbb{R}: w>0\}$ the sets of positive and strictly positive real numbers, respectively. The conjugate and the conjugate transpose of the vector $\mathbf{w}\in \mathbb{C}^{N}$ are denoted as $\overline{\mathbf{w}}\in \mathbb{C}^{N}$ and $\mathbf{w}^{H}\in \mathbb{C}^{N}$, respectively. The $n$th entry of a vector $\mathbf{w}$, which is assumed to be periodic, is written as $\mathbf{w}[n]$. We denote by $\tilde{\mathbf{w}}$ and $\hat{\mathbf{w}}$ the Fourier transform of a vector and its conjugate reflected version (that is, $\hat{\mathbf{w}}[n] := \overline{\mathbf{w}}[-n]$). The notation $\text{diag}(\mathbf{W},\ell)$ refers to a column vector with entries $\mathbf{W}[j,(j+\ell)\mod N]$ for $j=0,\cdots,N-1$. For vectors, $\|\mathbf{w}\|_p$ is the $\ell_p$ norm. Additionally, we use $\odot,$ and $*$ for the Hadamard (point-wise) product, and convolution, respectively. Finally, $\mathbb{E}[\cdot]$ represents the expected value.

\section{Problem formulation}
\label{sec:background}
FROG is probably the most commonly-used approach for full characterization of ultrashort optical pulses due to its simplicity and good experimental performance \cite{trebino2012frequency}. Experimentally, a FROG apparatus produces a two-dimensional intensity diagram, also known as FROG trace, of an input pulse by interacting the pulse with delayed versions of itself in a nonlinear-optical medium, usually using a second harmonic generation (SHG) crystal \cite{trebino2012frequency}. Mathematically, the FROG trace of a signal $\mathbf{x}\in \mathbb{C}^{N}$ is defined as
\begin{align}
\mathbf{Z}[p,k] &:= \left\lvert\sum_{n=0}^{N-1} \mathbf{x}[n]\mathbf{x}[n+pL]e^{-2\pi ink/N} \right \rvert^{2}, \nonumber\\
k&=0,\cdots,N-1,\hspace{0.3em}p=0,\cdots,R-1,
\label{eq:system}
\end{align}
with $R = \lceil N/L \rceil$ where $L<N$ and $i:=\sqrt{-1}$. This work assumes that the signal $\mathbf{x}$ is periodic, that is, $\mathbf{x}[n] = \mathbf{x}[n+lN]$ for any $l\in \mathbb{Z}$.

The FROG trace defined in \eqref{eq:system} can be considered as a map $\mathbb{C}^{N}\rightarrow \mathbb{R}_{+}^{\lceil N/L\rceil}$ that has three types of symmetry, usually called \textit{trivial ambiguities} in the PR literature. These ambiguities are summarized in Proposition \ref{theo:ambiguities}, using the following definition of a bandlimited signal.

\begin{definition}
	We say that $\mathbf{x}\in \mathbb{C}^{N}$ is a $\hspace{10em}B-bandlimited$ signal if its Fourier transform $\tilde{\mathbf{x}}\in \mathbb{C}^{N}$ contains $N-B$ consecutive zeros. That is, there exists $k$ such that $\tilde{\mathbf{x}}[k]=\cdots=\tilde{\mathbf{x}}[N+k+B-1]=0$.
	\label{def:bandlimitedSignal}
\end{definition}

\begin{proposition}(\cite{bendory2017signal})
	Let $\mathbf{x}\in \mathbb{C}^{N}$ be the underlying signal and let $\tilde{\mathbf{x}}\in \mathbb{C}^{N}$ be its Fourier transform. Let $\mathbf{Z}[p,k]$ be the FROG trace of $\mathbf{x}$ defined as in \eqref{eq:system} for some fixed $L$. Then, the following signals have the same FROG trace as $\mathbf{x}$:
	\begin{enumerate}
		\item the rotated signal $\mathbf{x}e^{i\phi}$ for some $\phi\in \mathbb{R}$;
		\item the translated signal $\mathbf{x}^{\ell}$ obeying $\mathbf{x}^{\ell}[n]=\mathbf{x}[n-\ell]$ for some $\ell \in \mathbb{Z}$ (equivalently, a signal with Fourier transform $\tilde{\mathbf{x}}^{\ell}$ obeying $\tilde{\mathbf{x}}^{\ell}[k] = \tilde{\mathbf{x}}[k]e^{-2\pi i\ell k/N}$ for some $\ell \in \mathbb{Z}$);
		\item the reflected signal $\hat{\mathbf{x}}$ obeying $\hat{\mathbf{x}}[n] := \overline{\mathbf{x}}[-n]$.
	\end{enumerate}
	If $\mathbf{x}$ is a $B$-bandlimited signal for some $B\leq N/2$, then the translation ambiguity is continuous. Namely, any signal with a Fourier transform such that $\tilde{\mathbf{x}}^{\psi}[k] := \tilde{\mathbf{x}}[k]e^{i\psi k}$ for some $\psi \in \mathbb{R}$, has the same FROG trace as $\mathbf{x}$.
	\label{theo:ambiguities}
\end{proposition}

Our goal is to estimate the signal $\mathbf{x}$, up to trivial ambiguities, from the FROG trace $\mathbf{Z}$. The work \cite{bendory2017signal} established that the pulse $\mathbf{x}$ can be uniquely identified (up to trivial ambiguities) from the FROG trace under rather mild conditions as summarized in the following proposition.

\begin{proposition}(\cite{bendory2017signal})
	Let $\mathbf{x}\in \mathbb{C}^{N}$ be a $B$-bandlimited signal as in Definition \ref{def:bandlimitedSignal} for some $B\leq N/2$. If $N/L\geq 4$, then almost all signals are determined uniquely from their FROG trace $\mathbf{Z}[p,k]$, up to trivial ambiguities, from $m\geq 3B$ measurements. If in addition we have access to the signal's power spectrum and $N/L \geq 3$, then $m\geq2B$ measurements suffice.
	\label{theo:uniqueness}
\end{proposition}
Proposition \ref{theo:uniqueness} has been recently extended to the case of blind ptychography, or blind FROG, in which the goal is to estimate two signals simultaneously \cite{bendory2018blind}. Evidently, Proposition \ref{theo:uniqueness} allows choices of $L>1$ meaning that not all the delay steps are needed to recover the pulse, and therefore a method that works in this regime as well is desired.

To take the ambiguities into account, we measure the relative error between the true signal $\mathbf{x}$ and any $\mathbf{w}\in \mathbb{C}^{N}$ as
\begin{equation}
\text{dist}(\mathbf{x},\mathbf{w}):= \frac{\left\lVert \sqrt{\mathbf{Z}}-\sqrt{\mathbf{W}} \right\rVert_{\text{F}}}{\left\lVert \sqrt{\mathbf{Z}} \right\rVert_{\text{F}}},
\label{eq:distance}
\end{equation}
where $\mathbf{Z}$ is the FROG trace of $\mathbf{x}$ according to \eqref{eq:system}, $\sqrt{\cdot}$ is the point-wise square root, $\mathbf{W}$ is the FROG trace of $\mathbf{w}$, and $\lVert \cdot \rVert_{\text{F}}$ denotes the Frobenius norm. Note that if $\text{dist}(\mathbf{x},\mathbf{w})=0$, and the uniqueness conditions of Proposition \ref{theo:uniqueness} are met, then for almost all signals $\mathbf{w}$ is equal to $\mathbf{x}$ up to trivial ambiguities.

In recent years, many papers have examined the problem of recovering a signal from phaseless quadratic random measurements. A popular approach is to minimize the intensity least-squares objective; see for instance \cite{candWir}. Recent works have shown that minimizing the amplitude least-squares objective leads to better reconstruction under noisy scenarios \cite{pauwels2018fienup,wang2016solving,zhang2016reshaped}. However, the latter cost function is non-smooth and thus may lead to a biased descent direction \cite{8410803}. To overcome the non-smoothness of the objective function, we follow the smoothing strategy proposed in \cite{8410803}. 

The smooth objective to recover the underlying pulse considered in this work is
\begin{align}
	\min_{\mathbf{z}\in \mathbb{C}^{n}} h(\mathbf{z},\mu) = \min_{\mathbf{z}\in \mathbb{C}^{n}} \frac{1}{NR}\sum_{k=0}^{N-1}\sum_{p=0}^{R-1} \ell_{k,p}(\mathbf{z},\mu),
	\label{eq:auxproblem}
\end{align}
where
\begin{equation}
	\small{\ell_{k,p}(\mathbf{z},\mu) := \left\lbrack\varphi_{\mu}\left(\left\lvert\sum_{n=0}^{N-1} \mathbf{z}[n]\mathbf{z}[n+pL]e^{-2\pi ink/N} \right \rvert\right) -\sqrt{\mathbf{Z}[p,k]}\right\rbrack^{2}}.
	\label{eq:ellfunction}
\end{equation}
The function $\varphi_{\mu}: \mathbb{R}\rightarrow \mathbb{R}_{++}$ in \eqref{eq:ellfunction} is defined as $$\varphi_{\mu}(w) := \sqrt{w^2+\mu^{2}},$$ with $\mu\in \mathbb{R}_{++}$ (a tunable parameter). Notice that if $\mu=0$, then \eqref{eq:ellfunction} reduces to the non-smooth formulation. In \cite{wang2016solving}, the authors addressed the non-smoothness by introducing truncation parameters into the gradient step in order to eliminate the errors in the estimated descent direction. However, this procedure can modify the search direction and increase the sample complexity of the phase retrieval problem \cite{8410803}.

In this work we propose a block stochastic gradient algorithm (BSGA) to solve \eqref{eq:auxproblem}, that is initialized by a spectral procedure which requires only a few iterations. Section \ref{sec:algorithm} explains in detail the proposed algorithm.

\section{Reconstruction Algorithm}
\label{sec:algorithm}
In order to solve the optimization problem in \eqref{eq:auxproblem}, we develop a gradient-based algorithm, called BSGA. The algorithm is initialized by the outcome of a spectral method approximating the signal $\mathbf{x}$ which will be explained in Section~\ref{sec:initialization}. 

To refine the initial estimate we use the Wirtinger derivatives as introduced in \cite{hunger2007introduction}. Let us define the vector $\mathbf{f}_{k}^{H}$ as 
\begin{equation}
\mathbf{f}_{k}^{H} := \left[\omega^{-0(k-1)},\omega^{-1(k-1)},\cdots,\omega^{-(n-1)(k-1)}\right],
\label{eq:vectora}
\end{equation}
with $\omega=e^{\frac{2\pi i}{n}}$ the $n$th root of unity. Then, the Wirtinger derivative of $h(\mathbf{z},\mu)$ in \eqref{eq:auxproblem} with respect to $\overline{\mathbf{z}}[\ell]$ is given by
\begin{align}
\hspace{-0.3em}\frac{\partial h(\mathbf{z},\mu)}{\partial \overline{\mathbf{z}}[\ell]} :=  &\frac{1}{NR}\sum_{k=0}^{N-1}\sum_{p=1}^{R-1}\left( \mathbf{f}_{k}^{H}\mathbf{g}_{p} - \upsilon_{k,p} \right) \overline{q}_{\ell,p}e^{2\pi i\ell k/N},
\label{eq:stochStep}
\end{align}
where $\upsilon_{k,p}:= \sqrt{\mathbf{Z}[p,k]}\frac{\mathbf{f}_{k}^{H}\mathbf{g}_{p}}{\varphi_{\mu}\left(\left\lvert\mathbf{f}_{k}^{H}\mathbf{g}_{p} \right \rvert\right)}$, and
\begin{align}
\overline{q}_{\ell,p}:=&\overline{\mathbf{z}}[\ell+p]+\overline{\mathbf{z}}[\ell-p]e^{-2\pi ikp/N},\nonumber\\
\mathbf{g}_{p} :=& \left[\mathbf{z}[0]\mathbf{z}[pL],\cdots,\mathbf{z}[N-1]\mathbf{z}[N-1+pL] \right]^{T}.
\end{align}
The gradient of $h(\mathbf{z},\mu)$ is then
\begin{align}
\frac{\partial h(\mathbf{z},\mu)}{\partial \overline{\mathbf{z}}} := \left[\frac{\partial h(\mathbf{z},\mu)}{\partial \overline{\mathbf{z}}[0]},\cdots,\frac{\partial h(\mathbf{z},\mu)}{\partial \overline{\mathbf{z}}[N-1]}\right]^{H}.
\label{eq:gradient}
\end{align}
Using \eqref{eq:gradient}, we define a standard gradient algorithm, taking the form of
\begin{align}
\mathbf{x}^{(t+1)}:=\mathbf{x}^{(t)} - \alpha \frac{\partial h(\mathbf{x}^{(t)},\mu^{(t)})}{\partial \overline{\mathbf{z}}},
\label{eq:updateforx}
\end{align}
where $\alpha$ is the step size. 

To alleviate the memory requirements and computational complexity required for large $N$, we suggest a block stochastic gradient descent strategy. Instead of calculating \eqref{eq:stochStep}, we choose only a random subset of the sum for each iteration $t$, that is,  
\begin{align}
\mathbf{d}_{\Gamma_{(t)}}[\ell]=\sum_{p,k\in \Gamma_{(t)}}\left( \mathbf{f}_{k}^{H}\mathbf{g}^{(t)}_{p} - \upsilon_{k,p,t} \right) \overline{q}^{(t)}_{\ell,p}e^{2\pi i\ell k/N},
\label{eq:functionh}
\end{align}
where the set $\Gamma_{(t)}$ is chosen uniformly and independently at random at each iteration $t$ from subsets of $\hspace{5em}\{1,\cdots,N \}\times \{1\cdots,R\}$ with cardinality $Q$. Specifically, the gradient in \eqref{eq:gradient} is uniformly sampled using a minibatch of data, in this case of size $Q$ for each step update, such that in expectation is \eqref{eq:stochStep} \cite[page 130]{spall2005introduction}.

As mentioned in Section \ref{sec:algorithm}, choosing $\mu>0$ prevents bias in the update direction. Since the function $h$ is smooth, we are able to construct a descent rule for $\mu$ (Line 13 of Algorithm \ref{alg:smothing}) in order to guarantee convergence to a first-order optimal point, that is, a point with zero gradient, in the vicinity of the solution.

\begin{algorithm}[h]
	\caption{(BSGA) Recovery from the FROG trace}
	\label{alg:smothing}
	\small
	\begin{algorithmic}[1]
		\State{\textbf{Input: }Data $\left\lbrace\mathbf{Z}[p,k]:k=0,\cdots,N-1,p=0,\cdots,R -1 \right\rbrace$. Choose constants $\gamma_{1},\gamma,\alpha\in(0,1)$, $\mu^{(0)}\geq 0$, cardinality $Q\in \{1,\cdots,NR\}$, and tolerance $\epsilon>0$.}
		\Statex{}
		\If{$L=1$}
		\State {Initial point $\mathbf{x}^{(0)} \leftarrow$ Algorithm 2$\left(\mathbf{Z}[p,k],T\right)$.}
		\Else
		\State {Initial point $\mathbf{x}^{(0)} \leftarrow$ Algorithm 3$\left(\mathbf{Z}[p,k],T\right)$.}
		\EndIf
		\Statex{}
		\While{$\left\lVert\mathbf{d}_{\Gamma_{(t)}}\right\rVert_{2}\geq \epsilon$}
		\Statex{Choose $\Gamma_{(t)}$ uniformly at random from the subsets of $\{1,\cdots,N \}\times \{1\cdots,R\}$ with cardinality $Q$ per iteration $t\geq 0$.}
		\State{$\displaystyle\mathbf{x}^{(t+1)}=\mathbf{x}^{(t)} - \alpha \mathbf{d}_{\Gamma_{(t)}} $,}
		\Statex{where}
		\State{$\displaystyle \mathbf{d}_{\Gamma_{(t)}}[\ell]=\sum_{p,k\in \Gamma_{(t)}}\left( \mathbf{f}_{k}^{H}\mathbf{g}^{(t)}_{p} - \upsilon_{k,p,t} \right) \overline{q}^{(t)}_{\ell,p}e^{2\pi i\ell k/N}.$}
		\State{$\upsilon_{k,p,t}= \sqrt{\mathbf{Z}[p,k]}\frac{\mathbf{f}_{k}^{H}\mathbf{g}^{(t)}_{p}}{\varphi_{\mu^{(t)}}\left(\left\lvert\mathbf{f}_{k}^{H}\mathbf{g}^{(t)}_{p} \right \rvert\right)} $.}
		\Statex{}
		\State{$\mathbf{g}^{(t)}_{p} = \left[\mathbf{x}^{(t)}[0]\mathbf{x}^{(t)}[pL],\cdots,\mathbf{x}^{(t)}[N-1]\mathbf{x}^{(t)}[N-1+pL] \right]^{T}$.}
		\Statex{}
		\State{$q^{(t)}_{\ell,p} = \mathbf{x}^{(t)}[\ell+p]+\mathbf{x}^{(t)}[\ell-p]e^{2\pi i kp/N}$.}
		\Statex{}
		\If{$\displaystyle\left\lVert\mathbf{d}_{\Gamma_{(t)}}\right\rVert_{2}\geq \gamma \mu^{(t)}$}
		\State{$\mu^{(t+1)}=\mu^{(t)}$.}
		\Else
		\State{$\mu^{(t+1)} = \gamma_{1}\mu^{(t)}$.}
		\EndIf
		\EndWhile
		\State{\textbf{return: } $\mathbf{x}^{(T)}$.}
	\end{algorithmic}
\end{algorithm}

\begin{theorem} 
Let $\mathbf{x}$ be $B$-bandlimited for some $B\leq N/2$, satisfying $\text{dist}(\mathbf{x},\mathbf{x}^{(t)})\leq \rho$ for some sufficiently small constant $\rho>0$. Suppose that $L=1$ and $\Gamma_{(t)}$ is sampled uniformly at random from all subsets of $\{1,\cdots,N \}\times \{1\cdots,R\}$ with cardinality $Q$, independently for each iteration. Then for almost all signals, Algorithm \ref{alg:smothing} with step size $\alpha\in (0,\frac{2}{U}] $ satisfies
	\begin{align}
	\lim_{t\rightarrow \infty} \mu^{(t)} = 0, \text{ and } \lim_{t\rightarrow \infty}\left\lVert\frac{\partial h(\mathbf{x}^{(t)},\mu^{(t)})}{\partial \overline{\mathbf{z}}}\right\rVert_{2}=0,
	\label{eq:convergence}
	\end{align}
	for some constant $U>0$ depending on $\rho$.
	\label{theo:contraction}
\end{theorem}
\begin{proof}
	See Appendix \ref{app:prooftheo4}.
\end{proof}
%

\section{Initialization Algorithm}
\label{sec:initialization}
In this section we devise a method to initialize the gradient iterations. This strategy approximates the signal $\mathbf{x}$ from the FROG trace as the leading eigenvector of a carefully designed matrix. We divide the exposition of the initialization procedure into two cases, $L=1$ and $L>1$, explained in Sections \ref{sub:initializationStep} and \ref{sub:initiLg1}, respectively.

\subsection{Initialization for $L=1$}
\label{sub:initializationStep}
Instead of directly dealing with the FROG trace in \eqref{eq:system}, we consider the acquired data in a transformed domain by taking its 1D DFT with respect to the frequency variable (normalized by $1/N$). Our measurement model is then
\begin{align}
&\mathbf{Y}[p,\ell] = \frac{1}{N}\sum_{k=0}^{N-1}\mathbf{Z}[p,k]e^{-2\pi i k\ell/N} \nonumber\\
=& \frac{1}{N}\sum_{k,n,m=0}^{N-1}\mathbf{x}[n]\overline{\mathbf{x}}[m]\mathbf{x}[n+pL]\overline{\mathbf{x}}[m+pL]e^{-2\pi ik\frac{(m-n-\ell)}{N}}\nonumber\\
=&\sum_{n=0}^{N-1}\mathbf{x}[n]\overline{\mathbf{x}}[n+\ell]\mathbf{x}[n+pL]\overline{\mathbf{x}}[n+\ell+pL],
\label{eq:system1}
\end{align}
where $p,\ell=0,\cdots,N-1$. Observe that for fixed $p$, $\mathbf{Y}[p,\ell]$ is the autocorrelation of $\mathbf{x}\odot\mathbf{x}_{pL}$, where  $\mathbf{x}_{pL}[n]=\mathbf{x}[n+pL]$.

Let $\mathbf{D}_{pL}\in \mathbb{C}^{N\times N}$ be a diagonal matrix composed of the entries of $\mathbf{x}_{pL}$, and let $\mathbf{C}_{\ell}$ be a circulant matrix that shifts the entries of a vector by $\ell$ locations, namely, $(\mathbf{C}_{\ell}\mathbf{x})[n]=\mathbf{x}[n+\ell]$. Then, the matrix $\mathbf{X}:=\mathbf{x}\mathbf{x}^{H}$ is linearly mapped to $\mathbf{Y}[p,\ell]$ as follows:
\begin{align}
\mathbf{Y}[p,\ell] &= \left(\mathbf{D}_{pL+\ell}\overline{\mathbf{D}}_{pL}\mathbf{C}_{\ell}\mathbf{x}\right)^{H}\mathbf{x} = \mathbf{x}^{H}\mathbf{A}_{p,\ell}\mathbf{x}\nonumber\\
&=tr(\mathbf{X}\mathbf{A}_{p,\ell}),
\label{eq:systemIniti}
\end{align}
where $\mathbf{A}_{p,\ell} = \mathbf{C}_{-\ell}\mathbf{D}_{pL}\overline{\mathbf{D}}_{pL+\ell}$, and $tr(\cdot)$ denotes the trace function. Observe that $\mathbf{C}_{\ell}^{T}=\mathbf{C}_{-\ell}$. Thus, we have that
\begin{align}
\mathbf{y}_{\ell} = \mathbf{G}_{\ell}\mathbf{x}_{\ell},
\label{eq:initiFinal}
\end{align}
for a fixed $\ell\in \{0,\cdots,N-1\}$, where $\mathbf{y}_{\ell}[n]=\mathbf{Y}[n,\ell]$ and $\mathbf{x}_{\ell} = \text{diag}(\mathbf{X},\ell)$. The $(p,n)$th entry of the matrix $\mathbf{G}_{\ell}\in \mathbb{C}^{\lceil\frac{N}{L} \rceil\times N}$ is given by 
\begin{align}
\mathbf{G}_{\ell}[p,n] := \mathbf{x}_{pL}[n]\overline{\mathbf{x}}_{pL}[n+\ell].
\label{eq:matrixG}
\end{align}
Since $L=1$, it follows from \eqref{eq:matrixG} that $\mathbf{G}_{\ell}$ is a circulant matrix. Therefore, $\mathbf{G}_{\ell}$ is invertible if and only if the DFT of its first column, in this case $\mathbf{x}\odot (\mathbf{C}_{\ell}\overline{\mathbf{x}})$, is non-vanishing.

Using \eqref{eq:initiFinal}, we propose a method to estimate the signal $\mathbf{x}$ from measurements \eqref{eq:system} using an alternating scheme: fixing $\mathbf{G}_{\ell}$, solving for $\mathbf{x}_{\ell}$, updating $\mathbf{G}_{\ell}$ and so forth. The new methodology proposed in \cite{bendory2018non} cannot be directly employed since here the matrices $\mathbf{G}_{\ell}$ are also unknown. Thus, our approach estimates the matrices $\mathbf{G}_{\ell}$ together with $\mathbf{x}_{\ell}$.

We start the alternating scheme by the initialization suggested in \cite{sidorenko2016ptychographic}
\begin{align}
\mathbf{x}_{ini\_pty}[r]:=\mathbf{v}[r] \exp(i\boldsymbol{\theta}[r]),
\label{eq:initialvector}
\end{align}
where $\boldsymbol{\theta}[r]\in [0,2\pi)$ is chosen uniformly at random for all $r\in \{0,\cdots,N-1\}$. The $r$th entry of $\mathbf{v}$ corresponds to the summation of the measured FROG trace over the frequency axis:
\begin{align}
\mathbf{v}[r] &:= \frac{1}{N}\sum_{k=0}^{N-1}\mathbf{Z}[r,k] = \sum_{k=0}^{N-1}\left\lvert\sum_{n=0}^{N-1} \mathbf{x}[n]\mathbf{x}[n+rL]e^{-2\pi ink/N} \right \rvert^{2} \nonumber\\
&:=\sum_{n=0}^{N-1} \lvert\mathbf{x}[n]\rvert^{2}\lvert \mathbf{x}[n+rL]\rvert^{2}.
\label{eq:vectorv}
\end{align}
Once the vector $\mathbf{x}_{ini\_ pty}$ is constructed, the vectors $\mathbf{x}_{\ell}^{(t)}$ at $t=0$ can be built as 
\begin{align}
	\mathbf{x}_{\ell}^{(0)}=\text{diag}(\mathbf{X}_{0}^{(0)},\ell),
	\label{eq:initialguess}
\end{align}
where 
\begin{equation}
	\mathbf{X}_{0}^{(0)}=\mathbf{x}_{ini\_pty}\mathbf{x}_{ini\_pty}^{H}.
	\label{eq:X0}
\end{equation}
Then, from \eqref{eq:initialguess} we proceed with an alternating procedure between estimating the matrix $\mathbf{G}_{\ell}$, and updating the vector $\mathbf{x}_{\ell}$ as follows.

\begin{itemize}[leftmargin=*]

	\item \textit{Update rule for $\mathbf{G}_{\ell}$: }In order to update $\mathbf{G}_{\ell}$, we update the matrix $\mathbf{X}_{0}^{(t)}$ as
	\begin{align}
	\text{diag}(\mathbf{X}_{0}^{(t)},\ell) = \mathbf{x}_{\ell}^{(t)}.
	\label{eq:solZ0}
	\end{align}
	
	Observe that if $\mathbf{x}_\ell^{(t)}$ is close to $\mathbf{x}_\ell$ for all $\ell$, then $\mathbf{X}_0^{(t)}$ is close to $\mathbf{x}\mathbf{x}^H$. Letting $\mathbf{w}^{(t)}$ be the leading (unit-norm) eigenvector of the matrix $\mathbf{X}_{0}^{(t)}$ constructed in \eqref{eq:solZ0}, from \eqref{eq:matrixG} each matrix $\mathbf{G}^{(t)}_{\ell}$ at iteration $t$ is given by
	\begin{align}
	\mathbf{G}_{\ell}^{(t)}[p,n] = \mathbf{x}_{pL}^{(t)}[n]\overline{\mathbf{x}}^{(t)}_{pL}[n+\ell],
	\label{eq:matrixG1}
	\end{align}
	where $\mathbf{x}^{(t)}_{pL}[n]= \mathbf{w}^{(t)}[n+pL]$. \vspace{0.5em}
	
	\item \textit{Optimization with respect to $\mathbf{x}_{\ell}$: }Fixing $\mathbf{G}_{\ell}^{(t-1)}$, one can estimate $\mathbf{x}^{(t)}_{\ell}$ at iteration $t$ by solving the linear least-squares (LS) problem 
	\begin{align}
	\min_{\mathbf{p}_{\ell}\in \mathbb{C}^{N}} & \hspace{1em} \lVert \mathbf{y}_{\ell}-\mathbf{G}^{(t-1)}_{\ell}\mathbf{p}_{\ell} \rVert_{2}^{2}.
	\label{eq:problemInitiz}
	\end{align}
	The relationship between the vectors $\mathbf{x}^{(t)}_{\ell}$ is ignored at this stage. If $\mathbf{G}_{\ell}^{(t-1)}$ is invertible, then the solution to this problem is given by $(\mathbf{G}_{\ell}^{(t-1)})^{-1}\mathbf{y}_{\ell}$. Since $\mathbf{G}_{\ell}^{(t-1)}$ is a circulant matrix, it is invertible if and only if the DFT of $\mathbf{x}^{(t-1)}\odot (\mathbf{C}_{\ell}\overline{\mathbf{x}}^{(t-1)})$ is non-vanishing. This condition cannot be ensured in general. Thus, we propose a surrogate proximal optimization problem to estimate $\mathbf{x}_{\ell}^{(t)}$ by
	\begin{align}
	\min_{\mathbf{p}_{\ell}\in \mathbb{C}^{N}} & \hspace{1em} \lVert \mathbf{y}_{\ell}-\mathbf{G}^{(t-1)}_{\ell}\mathbf{p}_{\ell} \rVert_{2}^{2} + \frac{1}{2\lambda}\lVert \mathbf{p}_{\ell} - \mathbf{x}_{\ell}^{(t-1)} \rVert_{2}^{2},
	\label{eq:problemInitiz1}
	\end{align}
	where $\lambda>0$ is a regularization parameter. In practice $\lambda$ is a tunable parameter \cite{parikh2014proximal}. In particular, for this work the value of $\lambda$ was determined using a cross-validation strategy such that each simulation uses the value that results in the smallest relative error according to \eqref{eq:distance}. The surrogate optimization problem in \eqref{eq:problemInitiz1} is strongly convex~\cite{parikh2014proximal}, and admits the following closed form solution
	\begin{align}
	\mathbf{x}_{\ell}^{(t)} = \mathbf{B}_{\ell,t}^{-1} \mathbf{e}_{\ell,t},
	\label{eq:solZl}
	\end{align}
	where
	\begin{align}
	\mathbf{B}_{\ell,t} &= \left(\mathbf{G}^{(t-1)}_{\ell}\right)^{H}\left(\mathbf{G}^{(t-1)}_{\ell}\right) + \frac{1}{2\lambda}\mathbf{I}, \nonumber\\
	\mathbf{e}_{\ell,t}&=\left(\mathbf{G}^{(t)}_{\ell}\right)^{H}\mathbf{y}_{\ell} + \frac{1}{2\lambda}\mathbf{x}^{(t-1)}_{\ell},
	\label{eq:auxMatrices}
	\end{align}
	with $\mathbf{I}\in \mathbb{R}^{N\times N}$ the identity matrix. Clearly $\mathbf{B}_{\ell,t}$ in \eqref{eq:auxMatrices} is always invertible. The update step for each $\mathbf{x}_{\ell}^{(t)}$ is computed in Line 9 of Algorithm \ref{alg:initialization}.
\end{itemize}\vspace{0.5em}

\begin{algorithm}[t!]
	\caption{Initialization Procedure $L=1$}
	\label{alg:initialization}
	\small
	\begin{algorithmic}[1]
		\State{\textbf{Input: }The measurements $\mathbf{Z}[p,k]$, $T$ the number of iterations, and $\lambda>0$. }
		\State{\textbf{Output:} $\mathbf{x}^{(0)}$ (estimation of $\mathbf{x}$).}
		\State{\textbf{Initialize: }$\mathbf{x}_{ini\_ pty}[r]=\mathbf{v}[r] \exp(i\boldsymbol{\theta}[r])$, and $\displaystyle\mathbf{v}[r]=\frac{1}{N}\sum_{k=0}^{N-1}\mathbf{Z}[r,k]$, $\boldsymbol{\theta}[r]\in [0,2\pi)$ is chosen uniformly and independently at random.}
		\State{Compute $\mathbf{Y}[p,\ell]$ the 1D inverse DFT with respect to $k$}
		\Statex{of $\mathbf{Z}[p,k]$.}
		\For{$t=1$ to $T$}
		\State{Construct $\mathbf{G}^{(t)}_{\ell}$ according to \eqref{eq:matrixG1}.}
		\State{Compute $\mathbf{B}_{\ell,t}= (\mathbf{G}^{(t)}_{\ell})^{H}(\mathbf{G}^{(t)}_{\ell}) + \frac{1}{2\lambda}\mathbf{I}$.}
		\State{Compute $\mathbf{e}_{\ell,t}=(\mathbf{G}^{(t)}_{\ell})^{H}\mathbf{y}_{\ell} + \frac{1}{2\lambda}\mathbf{x}^{(t-1)}_{\ell}$.}
		\State{Construct the matrix $\mathbf{X}^{(t)}_{0}$ such that
			\begin{equation*}
			\text{diag}(\mathbf{X}^{(t)}_{0},\ell)=
			\mathbf{B}_{\ell,t}^{-1} \mathbf{e}_{\ell,t}, \hspace{1em} \ell=0,\cdots,N-1.
			\end{equation*}}
		\State{Let $\mathbf{w}^{(t)}$ be the leading (unit-norm) eigenvector of $\mathbf{X}^{(t)}_{0}$.}
		\State{Take $\mathbf{x}^{(t)}_{pL}[n]= \mathbf{w}^{(t)}[n+pL]$.}
		\EndFor
		\State{Compute vector $\mathbf{x}^{(0)}$ as
			\begin{equation*}
			\mathbf{x}^{(0)} :=\sqrt{\sum_{n\in \mathcal{S}}\left(\mathbf{B}_{0,T}^{-1} \mathbf{e}_{0,T}\right)[n]}\mathbf{w}^{(T)},
			\end{equation*}
			where $\mathcal{S}:=\left\lbrace n: \left(\mathbf{B}_{0,T}^{-1} \mathbf{e}_{0,T}\right)[n]>0\right\rbrace$.}
		\State{\textbf{return: }$\mathbf{x}^{(0)}$.}
	\end{algorithmic}
\end{algorithm}

Finally, in order to estimate $\mathbf{x}$, the (unit-norm) principal eigenvector of $\mathbf{X}_{0}^{(T)}$ is normalized by
\begin{align}
	\beta=\sqrt{\sum_{n\in \mathcal{S}}\left(\mathbf{B}_{0,T}^{-1} \mathbf{e}_{0,T}\right)[n]},
	\label{eq:normestimated}
\end{align}
where $\mathcal{S}:=\left\lbrace n: \left(\mathbf{B}_{0,T}^{-1} \mathbf{e}_{0,T}\right)[n]>0\right\rbrace$. Observe that \eqref{eq:normestimated} results from the fact that $\sum_{n=0}^{N-1}\text{diag}(\mathbf{X},0)[n]=\lVert\mathbf{x}\rVert^{2}_{2}$. 

After a few iterations of this two-step procedure, the output is used to initialize the gradient algorithm described in Section \ref{sec:algorithm}. This alternating scheme is summarized in Algorithm \ref{alg:initialization}.


\subsection{FROG initialization step for $L>1$}
\label{sub:initiLg1}
Until now we focused on the case $L = 1$. If $L>1$, then the linear system in \eqref{eq:initiFinal} is underdetermined and $\mathbf{y}_{\ell}$ can be viewed as a subsampled version of \eqref{eq:system1} by a factor $L$. Therefore, in order to increase the number of equations when $L>1$, we up-sample $\mathbf{y}_{\ell}$ by a factor $L$. Specifically, we follow the proposed scheme in \cite{bendory2018non} that expands the measurement vector $\mathbf{y}_{\ell}$ by low-pass interpolation. Once the measurements are upsampled, we proceed as for $L = 1$. This initialization, for $L>1$, is summarized in Algorithm \ref{alg:initialization1}. From Line 3 to Line 5 the low-pass interpolation by a factor $L$ is computed, and then in Line 6, Algorithm \ref{alg:initialization} generates the initial estimation of the underlying signal.

\begin{algorithm}[H]
	\caption{Initialization Procedure $L>1$}
	\label{alg:initialization1}
	\small
	\begin{algorithmic}[1]
		\State{\textbf{Input: }The measurements $\mathbf{Z}[p,k]$, $T$ the number of iterations, and a smooth interpolation filter $\mathbf{s}_{L}$ that approximates a lowpass filter with bandwidth $\lceil N/L \rceil$.}
		\State{\textbf{Output:} $\mathbf{x}^{(0)}$ (estimation of $\mathbf{x}$).}
		\State{Compute $\mathbf{Y}[p,\ell]$ as the 1D DFT with respect to $k$}
		\Statex{of $\mathbf{Z}[p,k]$.}
		\State{\vspace{-1.2em}\begin{itemize}
				\item \textit{Expansion: }
				\begin{equation*}
					\grave{\mathbf{y}}_{\ell}[n] = \left \lbrace\begin{array}{ll}
					\mathbf{y}_{\ell}[p] & \text{ if } n=pL\\
					0& \text{ otherwise. }
					\end{array}\right.
				\end{equation*}
				\item \textit{Interpolation: }
				\begin{equation*}
					\mathbf{y}^{(I)}_{\ell} = \grave{\mathbf{y}}_{\ell}*\mathbf{s}_{L}.
				\end{equation*}
			\end{itemize}}
		\State{Compute $\mathbf{Y}^{(I)}[p,\ell]=\mathbf{y}^{(I)}_{\ell}[p]$}.
		\Statex{}
		\State{Compute $\mathbf{Z}^{(I)}[p,k] = \lvert \tilde{\mathbf{Y}}^{(I)}[p,k] \rvert^{2}$ where $\tilde{\mathbf{Y}}^{(I)}[p,k]$ is the 1D inverse DFT with respect to $\ell$ of $\mathbf{Y}^{(I)}[p,\ell]$.}
		\State{Compute $\mathbf{x}^{(0)}\leftarrow \text{Algorithm 2}(\mathbf{Z}^{(I)},T)$}
		\State{\textbf{return: }$\mathbf{x}^{(0)}$.}\vspace{0.4em}
	\end{algorithmic}
\end{algorithm}
%

\section{Numerical Results}
\label{sec:numerical}
This section evaluates the numerical performance of BSGA and compares the results with the stochastic gradient algorithm Ptych proposed in \cite{sidorenko2016ptychographic}. We used the following parameters for Algorithm \ref{alg:smothing}:  $\gamma_{1}=0.1$, $\gamma=0.1$, $\alpha = 0.6$, $\mu_{0}=65$, and $\epsilon=1\times 10^{-10}$. The number of indices that are chosen uniformly at random is fixed as $Q=N$. A cubic interpolation was used in Algorithm \ref{alg:initialization1} (see Line 4), and the regularization parameter was fixed to $\lambda=0.5$.

Five tests were conducted to evaluate the performance of the proposed method under noisy and noiseless scenarios at different values of signal-to-noise-ratio (SNR), defined as SNR$ = 10\log_{10}(\lVert \mathbf{Z} \rVert^{2}_{\text{F}}/\lVert \boldsymbol{\sigma} \rVert^{2}_{\text{2}})$, where $\boldsymbol{\sigma}$ is the variance of the noise. First, we examine the empirical success rate of BSGA for different values of $L$. The second experiment assesses the performance of the initialization technique and its impact on the reconstruction quality. Third, we show several examples of reconstructed pulses attained with BSGA and Ptych under noisy and noiseless scenarios, when the complete FROG trace is used. The fourth experiment investigates the performance of the proposed method and Ptych in reconstructing the pulses when $L>1$ and the FROG trace is corrupted by noise. The last test compares the computational complexity between the reconstruction methods in terms of their running time to reach a given relative error.

The signals used in the simulations were constructed as follows. For all tests, we built a set of $\left\lceil \frac{N-1}{2} \right\rceil$-bandlimited pulses that conform to a Gaussian power spectrum centered at 800 nm. Specifically, each pulse ($N=128$ grid points) is produced via the Fourier transform of a complex vector with a Gaussian-shaped amplitude with a cutoff frequency of $150$ femtoseconds$^{-1}$ (fsec$^{-1}$). Next, we multiply the obtained power spectrum by a uniformly distributed random phase. In the experiments we used the inverse Fourier of this signal as the underlying pulse.

All simulations were implemented in Matlab R2019a on an Intel Core i7 3.41Ghz CPU with 32 GB RAM. The code for BSGA is publicly available at \url{https://github.com/samuelpinilla/FROG}. The code of Ptych was downloaded from the authors' website\footnote{\url{https://oren.net.technion.ac.il/homepage/}}.

\subsection{Empirical Probability of Success}
This section numerically evaluates the success rate of BSGA. To this end, BSGA and Ptych are initialized at $\mathbf{x}^{(0)} = \mathbf{x}+\delta\zeta$, where $\delta$ is a fixed constant and $\zeta$ takes values on $\{-1,1\}$ with equal probability, while $L$ ranges from 1 to 6. A trial is declared successful when the returned estimate attains a relative error as in \eqref{eq:distance} that is smaller than $10^{-6}$. We numerically determine the empirical success rate among 100 trials. Fig. \ref{fig:orthoini} summarizes these results, and shows that BSGA performs better than Ptych, since it is able to retrieve the signal for larger values of $L$.

\begin{figure}[h]
	\centering
	\includegraphics[width=1\linewidth]{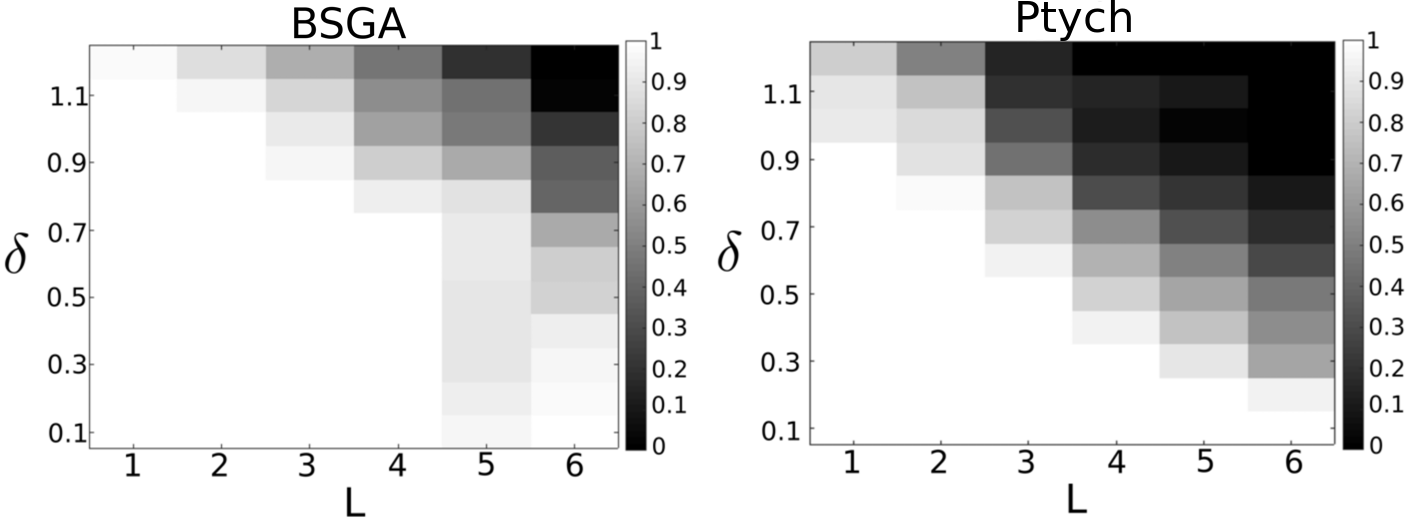}
	\caption{\footnotesize{Empirical success rate comparison between BSGA and Ptych as a function of $L$ and $\delta$ in the absence of noise.}}
	\label{fig:orthoini}
\end{figure}

\begin{figure}[t!]
	\centering
	\includegraphics[width=1\linewidth]{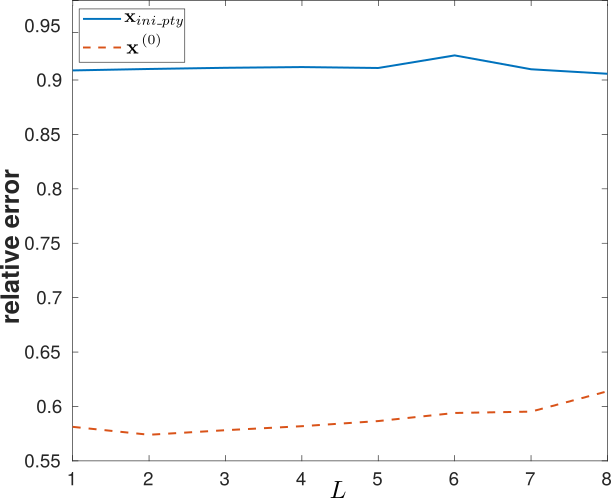}
	\caption{\footnotesize{Relative error comparison between the initial vector $\mathbf{x}_{ini\_ pty}$ as defined in \eqref{eq:initialvector}, and the returned initial guess $\mathbf{x}^{(0)}$ for different values of $L$ in the absence of noise. For each value of $L$, an average of the relative error was computed among 100 trials.}}
	\label{fig:initerrors}
\end{figure}
\begin{figure}[b!]
	\centering
	\includegraphics[width=1\linewidth]{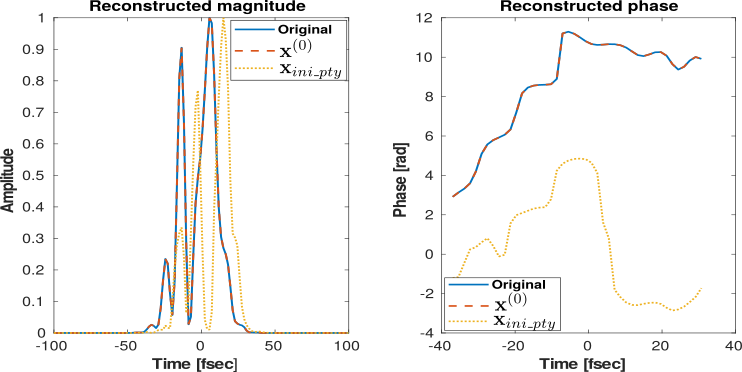}
	\caption{\footnotesize{Reconstructed pulses from the FROG trace with $L=4$ using Algorithm \ref{alg:smothing} initialized by $\mathbf{x}_{ini\_ pty}$ and the returned vector $\mathbf{x}^{(0)}$ using Algorithm~\ref{alg:initialization1}.}}
	\label{fig:randomInit}
	\vspace{-0.5em}
\end{figure}

\subsection{Relative Error of the Initialization Procedure}
This section examines the impact of the designed initialization described in Algorithms \ref{alg:initialization} and \ref{alg:initialization1}, under noisy and noiseless scenarios. We compare the relative error between the starting vector in \eqref{eq:initialvector}, and the returned solution $\mathbf{x}^{(0)}$ of the proposed initialization procedure. The number of iterations to attain the vector $\mathbf{x}^{(0)}$ using the designed initialization was fixed as $T=2$, and we numerically determine the relative error averaged over 100 trials. These numerical results are summarized in Fig. \ref{fig:initerrors}, and indicate that the proposed initialization algorithm outperforms $\mathbf{x}_{ini\_ pty}$.

In order to illustrate the effect of the initial guesses, we ran Algorithm \ref{alg:smothing} initialized by $\mathbf{x}_{ini\_ pty}$ and $\mathbf{x}^{(0)}$ with $L=4$. Fig. \ref{fig:randomInit} shows the attained reconstructions. Notice that the proposed reconstruction algorithm fails in estimating the input pulse when it was initialized by $\mathbf{x}_{ini\_ pty}$.

We numerically determine the performance of the proposed initialization at different SNR levels, with $L$ ranging from 1 to 8. Specifically, we added white noise to the FROG measurements at different SNR levels: SNR = 8dB, 12dB, 16dB and 20dB. Fig. \ref{fig:init} displays the relative error attained by the proposed initialization for different SNR and~$L$ values.

\begin{figure}[t!]
	\centering
	\includegraphics[width=0.9\linewidth]{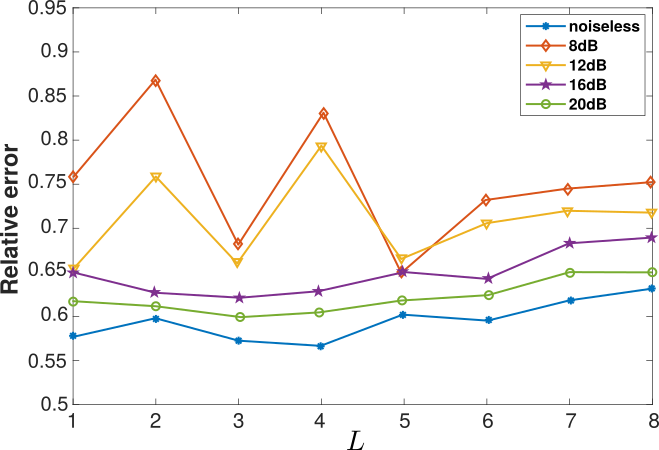}
	\caption{\footnotesize{Performance of the proposed initialization described in Algorithms \ref{alg:initialization} and \ref{alg:initialization1} at different SNR levels, with $L$ ranging from 1 to 8. For each value of $L$, the relative error was averaged over 100 trials.}}
	\label{fig:init}
\end{figure}

From Fig. \ref{fig:init} it can be seen that the returned initialization at levels of SNR $\leq 16$dB is, approximately, independent of the value of $L$ when $L\leq 6$. Combining these numerical results with Fig. \ref{fig:orthoini}, we conclude that BSGA is able to better estimate the underlying pulse (up to trivial ambiguities) if $L\leq4$ for both noiseless and noisy scenarios compared to Ptych.

Finally, we numerically determine the empirical success rate of BSGA with increasing $L$, in the absence of noise, when Algorithm \ref{alg:smothing} is initialized with $\mathbf{x}_{ini\_ pty}$, a random vector and $\mathbf{x}^{(0)}$. A trial is declared successful when the returned estimate attains a relative error as in \eqref{eq:distance} that is smaller than $1\times 10^{-6}$. The results are summarized in Fig. \ref{fig:success}, where the number of iterations that BSGA requires to reach the given relative error for $L=1$ is also presented. The success rate and the number of iterations are averaged over 100 pulses. The reported results show the effectiveness of Algorithm \ref{alg:smothing} when it is initialized by $\mathbf{x}^{(0)}$ for $L>1$.
\begin{figure}[h]
	\centering
	\includegraphics[width=0.8\linewidth]{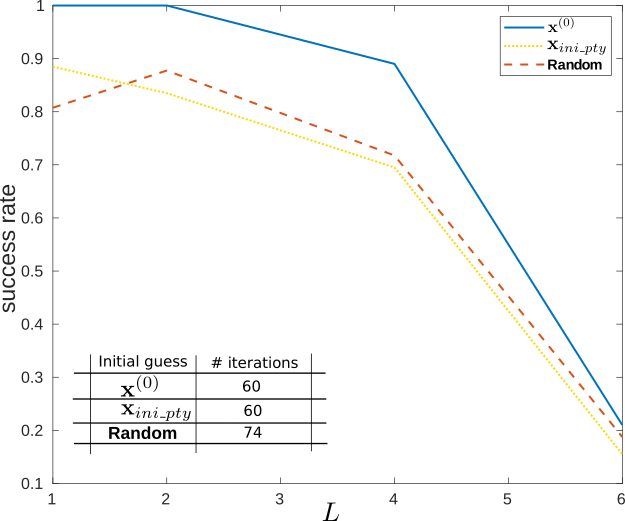}
	\caption{\footnotesize{Empirical success rate of Algorithm \ref{alg:smothing} when it is initialized by $\mathbf{x}^{(0)}$, $\mathbf{x}_{ini\_ pty}$ and a random vector as a function of $L$ in the absence of noise.}}
	\label{fig:success}
\end{figure}

\subsection{Pulse Reconstruction Examples for $L=1$}
In this section we show the performance of BSGA in recovering two pulses under noiseless and noisy scenarios for $L=1$. The results are presented in Figs. \ref{fig:plot1}, and \ref{fig:plot11}, respectively, where the attained relative errors by BSGA and Ptych are included. For the second scenario, the FROG trace is corrupted by Gaussian noise with SNR = 20dB.

\begin{figure}[h]
	\centering
	\includegraphics[width=1\linewidth]{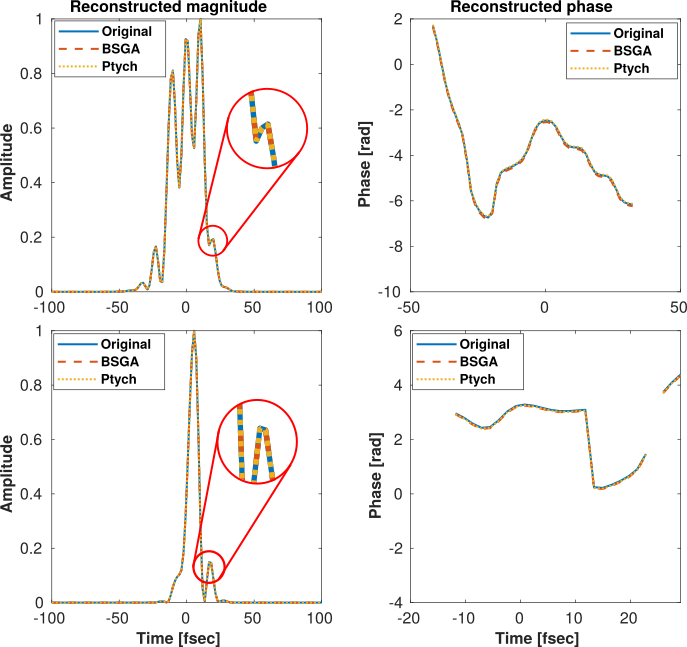}
	\caption{\footnotesize{Reconstructed pulses from complete FROG data ($L=1$), in the absence of noise. The attained error for both BSGA and Ptych was $1\times 10^{-6}$.}}
	\label{fig:plot1}
	\vspace{-1em}
\end{figure}

From the results in Fig. \ref{fig:plot1} it can be observed that both methods, BSGA and Ptych, are able to estimate the pulses and provide similar results for the noiseless case.

On the other hand, in Fig. \ref{fig:plot11}, the attained reconstructions, for the noisy scenario, indicate that BSGA is able to better estimate the pulse compared to Ptych. This advantage is obtained because of the effectiveness of the proposed smoothing update step and initialization strategy from complete data as reported in Fig. \ref{fig:orthoini}, and Figs. \ref{fig:initerrors}, \ref{fig:init}, respectively.
\begin{figure}[h]
	\centering
	\includegraphics[width=1\linewidth]{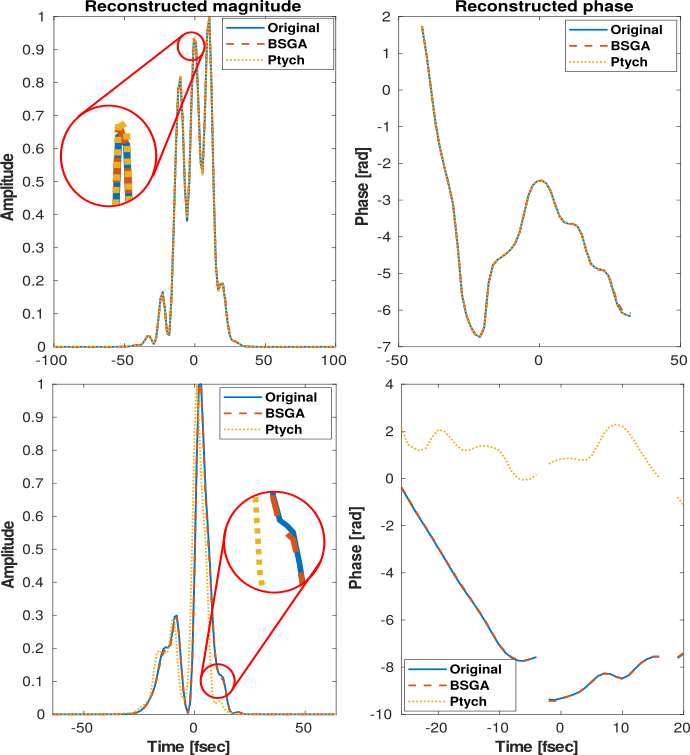}
	\caption{\footnotesize{Reconstructed pulses from complete noisy FROG data ($L=1$), with \hspace{7em}SNR = 20dB. The attained relative error for the top pulse for both BSGA and Ptych was $5\times 10^{-2}$. For the bottom pulse the attained errors were $5\times 10^{-2}$ and $2\times 10^{-1}$ for BSGA and Ptych respectively.}}
	\label{fig:plot11}
	\vspace{-1.5em}
\end{figure}

\begin{figure}[t!]
	\centering
	\includegraphics[width=0.8\linewidth]{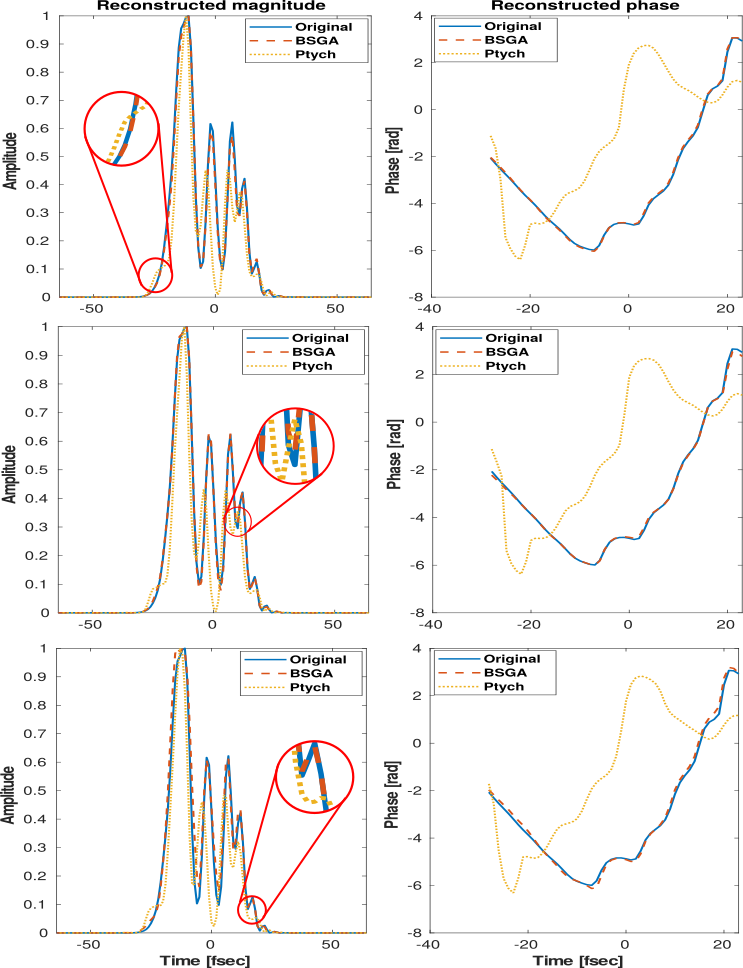}
	\caption{\footnotesize{Reconstructed pulses from incomplete noisy FROG traces \hspace{7em}(SNR = 20dB), for different values of $L$. (a) $L=2$, (b) $L=4$, and (c) $L=8$.}}
	\label{fig:plot21}
	\vspace{-2em}
\end{figure}

\subsection{Pulse Reconstruction Examples for $L>1$}
Next, we examine the recovery performance of BSGA from noisy incomplete data by adding Gaussian noise with SNR = 20dB, for $L\in \{2,4,8\}$. Figs. \ref{fig:plot21} and \ref{fig:plot2} illustrate the attained reconstructions for BSGA and Ptych; their attained relative errors are also reported in Fig. \ref{fig:plot2}. These figures suggest that BSGA better estimates the pulse and its FROG trace compared to Ptych over a range of values of $L$. This advantage is obtained because of the effectiveness of the proposed smoothing update step and initialization strategy from incomplete data as reported in Fig. \ref{fig:orthoini}, and Figs. \ref{fig:initerrors}, \ref{fig:init}, respectively.\vspace{-0.5em}

\begin{figure}[H]
	\centering
	\includegraphics[width=0.85\linewidth]{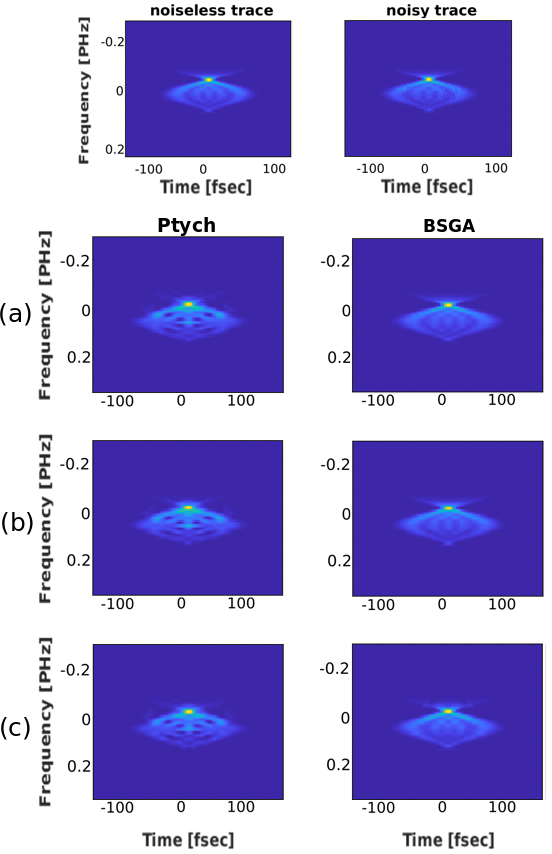}
	\caption{\footnotesize{Reconstruction of full FROG traces from incomplete noisy data for all methods. Top row shows the desirable full FROG trace without and with noise of SNR = 20dB. (a) $L=2$, (b) $L=4$, and (c) $L=8$. The attained errors for BSGA and Ptych were $5\times 10^{-2}$ and $2\times 10^{-1}$, respectively, for all the reconstructed FROG traces.}}
	\label{fig:plot2}
	\vspace{-1em}
\end{figure}

\subsection{Computational Complexity}
Simulations were conducted to compare the speed of convergence of the algorithms in the absence of noise, for $L=1,2$ and $4$. Table \ref{tab:time} reports the number of iterations and running time required by BSGA and Ptych to achieve a relative error of $1\times 10^{-6}$, averaged over 100 pulses. The experiment shows that BSGA is similar in time and number of iterations compared to Ptych for a range of values of $L$.\vspace{-0.5em}
\begin{center}
	\begin{table}[h]
		\caption{Comparison of iteration count and time cost}
		\centering
		\begin{tabular}{|c|c|c|c|}
			\hline
			&\textbf{Algorithms}& \textbf{Iterations} & \textbf{Time (s)}\\
			\hline
			\multirow{2}{*}{$L=1$}&BSGA & 60& 1.451 \\
			&Ptych &36 & 1.325 \\
			\hline
			\multirow{2}{*}{$L=2$}&BSGA & 111& 1.567 \\
			&Ptych &125 & 1.954 \\
			\hline
			\multirow{2}{*}{$L=4$}&BSGA & 265& 1.772 \\
			&Ptych &300 & 2.013 \\
			\hline
		\end{tabular}
		\label{tab:time}
	\end{table}
\vspace{-1em}
\end{center}

\section{Conclusion}
\label{sec:conclusion}
This paper presented a new method called BSGA, to fully characterize a pulse from its FROG trace. Our algorithm consists of two steps: a spectral initialization followed by successive refinements based upon a sequence of block stochastic gradient iterations. The principles of the proposed method were adopted from algorithms developed for the problems of recovering a signal from random quadratic measurements, and from the short-time Fourier phaseless data.

Numerical experiments were conducted to evaluate the performance of BSGA in comparison to the ptychography technique developed in \cite{sidorenko2016ptychographic}. The results show improvements in recovering the pulse for both magnitude and phase, from noisy incomplete data. Additionally, the numerical results suggest the effectiveness of the proposed initialization under both noiseless and noisy scenarios with incomplete data. Future work should include implementing BSGA on real data to further validate its performance. Another interesting research direction is to examine similar strategies for blind FROG in which two signals are estimated simultaneously \cite{trebino2012frequency}. 


\ifCLASSOPTIONcaptionsoff
  \newpage
\fi

\appendices

\section{Proof Theorem 1}
\label{app:prooftheo4}
Let us define the search set as
\begin{align}
\mathcal{J}:=\{ \mathbf{z}\in \mathbb{C}^{N}, \text{$B$-bandlimited}: \text{dist}(\mathbf{x},\mathbf{z})\leq \rho, B\leq N/2 \},
\label{eq:covergenceset}
\end{align}
for some small constant $\rho>0$. Recall that $\mathbf{z}$ is a $B$-bandlimited signal if there exists $k$ such that $\tilde{\mathbf{z}}[k]=\cdots=\tilde{\mathbf{z}}[N+k+B-1]=0$, where $\tilde{\mathbf{z}}$ is the Fourier transform of $\mathbf{z}$. The bandlimit condition guarantees that we have unique solution, according to Proposition~\ref{theo:uniqueness}. 

In order to prove Theorem \ref{theo:contraction}, the function $h(\mathbf{z},\mu)$ in \eqref{eq:auxproblem} needs to satisfy the four requirements stated in the following lemma, which are used in the analysis of convergence for stochastic gradient methods \cite{ghadimi2013stochastic}. 

\begin{lemma}
	\label{assu:assumption}
	The function $h(\mathbf{z},\mu)$ in \eqref{eq:auxproblem} and its Wirtinger derivative in \eqref{eq:gradient} satisfy the following properties.
	\begin{enumerate}
		\item The cost function $h(\mathbf{z},\mu)$ in \eqref{eq:auxproblem} is bounded below.
		
		\item The set $\mathcal{J}$ as defined in \eqref{eq:covergenceset} is closed and bounded.
		
		\item There exists a constant $U>0$, such that
		\begin{equation}
		\left\lVert \frac{\partial h(\mathbf{z}_{1},\mu)}{\partial \overline{\mathbf{z}}}  - \frac{\partial h(\mathbf{z}_{2},\mu)}{\partial \overline{\mathbf{z}}} \right\rVert_{2}\leq U\left\lVert\mathbf{z}_{1}-\mathbf{z}_{2} \right\rVert_{2},
		\label{eq:assumption}
		\end{equation}
		holds for all $\mathbf{z}_{1},\mathbf{z}_{2}\in \mathcal{J}$.
		
		\item For all $\mathbf{z}\in \mathcal{J}$
		\begin{align}
		\mathbb{E}_{\Gamma_{(t)}}\left[\left\lVert \mathbf{d}_{\Gamma_{(t)}}  - \frac{\partial h(\mathbf{z},\mu)}{\partial \overline{\mathbf{z}}} \right\rVert^{2}_{2}\right]\leq \zeta^{2},
		\label{eq:variance}
		\end{align}
		for some $\zeta>0$, where $\mathbf{d}_{\Gamma_{(t)}}$ is as in Line 9 of Algorithm \ref{alg:smothing}.
	\end{enumerate}
\end{lemma}
\begin{proof}
	See Appendix \ref{app:lemma}.
\end{proof}
%
%

To prove Theorem \ref{theo:contraction}, denote the set $\hspace{8em}\mathcal{K}_{1}:=\{t|\mu^{(t+1)}=\gamma_{1}\mu^{(t)} \}$ with $\gamma_{1}\in (0,1)$, which is a tunable parameter \cite{zhang2009smoothing}. If the set $\mathcal{K}_{1}$ is finite, then according to Lines 13-16 in Algorithm \ref{alg:smothing} there exists an integer $\grave{t}$, such that, for all $t>\grave{t}$
\begin{align}
\left\lVert \mathbf{d}_{\Gamma_{(t)}} \right\rVert_{2}\geq \gamma \mu^{(\grave{t})},
\end{align}
with $\gamma\in (0,1)$. Taking $\grave{\mu}=\mu^{(\grave{t})}$, the optimization problem \eqref{eq:auxproblem} reduces to
\begin{align}
\min_{\mathbf{z}\in \mathbb{C}^{N}} h(\mathbf{z},\grave{\mu}).
\label{eq:resultantOptiProblem}
\end{align}
Now, considering the properties stated in Lemma \ref{assu:assumption}, from \cite[Theorem 2.1]{ghadimi2013stochastic} we get 
\begin{align}
\lim_{t\rightarrow \infty} \left\lVert\frac{\partial h(\mathbf{x}^{(t)},\mu^{(t)})}{\partial \overline{\mathbf{z}}}\right\rVert_{2}=\lim_{t\rightarrow \infty} \left\lVert\mathbb{E}_{\Gamma_{(t)}}\left[\mathbf{d}_{\Gamma_{(t)}}\right] \right\rVert_{2}=0.
\label{eq:conjugate1}
\end{align}
It can be readily seen that \eqref{eq:conjugate1} contradicts the assumption $\left\lVert \mathbf{d}_{\Gamma_{(t)}} \right\rVert_{2}\geq \gamma \mu^{(\grave{t})}$, for all $t>\grave{t}$. This shows that $\mathcal{K}_{1}$ must be infinite and $\displaystyle\lim_{t\rightarrow \infty} \mu^{(t)} = 0$. 

Given that $\mathcal{K}_{1}$ is infinite, we deduce that

\small{\begin{align}
	\lim_{t\rightarrow \infty} \left\lVert\frac{\partial h(\mathbf{x}^{(t)},\mu^{(t)})}{\partial \overline{\mathbf{z}}}\right\rVert_{2}&=\lim_{t\rightarrow \infty} \left\lVert\mathbb{E}_{\Gamma_{(t)}}\left[\mathbf{d}_{\Gamma_{(t)}}\right] \right\rVert_{2} \nonumber\\
	&\leq \lim_{t\rightarrow \infty} \mathbb{E}_{\Gamma_{(t)}}\left[\left\lVert\mathbf{d}_{\Gamma_{(t)}}\right\rVert_{2}\right]  \leq \gamma \lim_{t\rightarrow \infty} \mu^{(t)} = 0,
	\label{eq:gradientMu4}
	\end{align}}\normalsize
where the second line follows from the Jensen inequality. Therefore, from \eqref{eq:gradientMu4} the result of Theorem \ref{theo:contraction} holds.

\section{Proof of Lemma \ref{assu:assumption}}
\label{app:lemma}
The proof of Lemma \ref{assu:assumption} is obtained by individually proving the following four requirements.

1) Following from the definition of $h(\mathbf{z},\mu)$ in \eqref{eq:auxproblem} it is clear that $h(\mathbf{z},\mu)\geq 0$ and thus bounded below.

2) This holds by definition.

3) From \eqref{eq:stochStep} it follows that the $\ell$-th entry of $\frac{\partial h(\mathbf{z},\mu)}{\partial \overline{\mathbf{z}}}$ is given by
\begin{align}
\frac{\partial h(\mathbf{z},\mu)}{\partial \overline{\mathbf{z}}}[\ell] = \frac{1}{N^{2}}\sum_{k,p=0}^{N-1}\left( \mathbf{f}_{k}^{H}\mathbf{g}_{p}(\mathbf{z}) - \upsilon_{k,p} \right) \overline{q}_{\ell,p}e^{2\pi i\ell k/N},
\end{align}
where $\upsilon_{k,p}= \sqrt{\mathbf{Z}[p,k]}\frac{\mathbf{f}_{k}^{H}\mathbf{g}_{p}(\mathbf{z})}{\varphi_{\mu}\left(\left\lvert\mathbf{f}_{k}^{H}\mathbf{g}_{p}(\mathbf{z}) \right \rvert\right)}$, and 
\begin{align}
\overline{q}_{\ell,p}&=\overline{\mathbf{z}}[\ell+p]+\overline{\mathbf{z}}[\ell-p]e^{-2\pi ikp/N},\nonumber\\
\mathbf{g}_{p}(\mathbf{z}) &= \left[\mathbf{z}[0]\mathbf{z}[pL],\cdots,\mathbf{z}[N-1]\mathbf{z}[N-1+pL] \right]^{T}\nonumber.
\label{eq:auxvector}
\end{align}
Let $\mathbf{D}_{p}(\mathbf{z})$ be a diagonal matrix composed of the entries of $\overline{\mathbf{z}}_{pL}[n]=\overline{\mathbf{z}}[n+pL]$. Using \eqref{eq:vectora}, the term $\overline{q}_{\ell,p}e^{2\pi i\ell k/N}$ can be rewritten as 
\begin{align}
\overline{q}_{\ell,p}e^{2\pi i\ell k/N}=\left(\mathbf{D}_{p}(\mathbf{z})\mathbf{f}_{k}\right)[\ell] + \omega^{-kp}\left(\mathbf{D}_{-p}(\mathbf{z})\mathbf{f}_{k}\right)[\ell].
\end{align}
Thus,
\begin{align}
\frac{\partial h(\mathbf{z},\mu)}{\partial \overline{\mathbf{z}}} = \frac{1}{N^{2}}\sum_{p,k=0}^{N-1}f_{k,p}(\mathbf{z}) + g_{k,p}(\mathbf{z}),
\label{eq:convergence3}
\end{align}
where
\begin{align}
	f_{k,p}(\mathbf{z})&=\rho_{k,p}(\mathbf{z})\mathbf{D}_{p}(\mathbf{z})\mathbf{f}_{k},\nonumber\\
	g_{k,p}(\mathbf{z})&=\omega^{-kp}\rho_{k,p}(\mathbf{z})\mathbf{D}_{-p}(\mathbf{z})\mathbf{f}_{k},
\end{align}
and
\begin{align}
\rho_{k,p}(\mathbf{z})= \mathbf{f}_{k}^{H}\mathbf{g}_{p}(\mathbf{z}) -\sqrt{\mathbf{Z}[p,k]}\frac{\mathbf{f}_{k}^{H}\mathbf{g}_{p}(\mathbf{z})}{\varphi_{\mu}\left(\left\lvert\mathbf{f}_{k}^{H}\mathbf{g}_{p}(\mathbf{z}) \right \rvert\right)}.
\label{eq:rhokp}
\end{align}

To prove 3) we establish that any $f_{k,p}(\mathbf{z})$ and $g_{k,p}(\mathbf{z})$ satisfy
\begin{align}
\left\lVert f_{k,p}(\mathbf{z}_{1}) -  f_{k,p}(\mathbf{z}_{2}) \right\rVert_{2} \leq r_{k,p} \lVert \mathbf{z}_{1} - \mathbf{z}_{2} \rVert_{2},
\label{eq:functionf}
\end{align}
and
\begin{align}
\left\lVert g_{k,p}(\mathbf{z}_{1}) -  g_{k,p}(\mathbf{z}_{2}) \right\rVert_{2} \leq s_{k,p} \lVert \mathbf{z}_{1} - \mathbf{z}_{2} \rVert_{2},
\label{eq:functiong}
\end{align}
for all $\mathbf{z}_{1},\mathbf{z}_{2}\in \mathcal{J}$ and some constants $r_{k,p},s_{k,p}>0$. In fact, once we prove \eqref{eq:functionf}, it can be performed a similar analysis for $g_{k,p}(\mathbf{z})$, and thus the result of this third part holds. 

From the definition of $f_{k,p}(\mathbf{z})$, for any $\mathbf{z}_{1},\mathbf{z}_{2}\in \mathcal{J}$ we have that
\begin{equation}
\small{\frac{1}{\sqrt{N}}\left\lVert f_{k,p}(\mathbf{z}_{1}) -  f_{k,p}(\mathbf{z}_{2}) \right\rVert_{2}\leq  \left\lVert \rho_{k,p}(\mathbf{z}_{1})\overline{\mathbf{z}}_{1} - \rho_{k,p}(\mathbf{z}_{2})\overline{\mathbf{z}}_{2} \right\rVert_{2},}
\label{eq:ine1}
\end{equation}
considering that $\mathbf{D}_{p}(\mathbf{z}_{1})$ and $\mathbf{D}_{p}(\mathbf{z}_{2})$ are diagonal matrices, and $\lVert \mathbf{f}_{k} \rVert_{2}=\sqrt{N}$. Observe that from \eqref{eq:rhokp} and \eqref{eq:ine1} it can be obtained that
\begin{align}
&\frac{1}{\sqrt{N}}\left\lVert f_{k,p}(\mathbf{z}_{1}) -  f_{k,p}(\mathbf{z}_{2}) \right\rVert_{2} \nonumber\\
\leq& \frac{\left\lvert\mathbf{f}_{k}^{H}\mathbf{g}_{p}(\mathbf{z}_{1})\right\rvert}{\mu}\left(\varphi_{\mu}\left(\left\lvert\mathbf{f}_{k}^{H}\mathbf{g}_{p}(\mathbf{z}_{1}) \right \rvert\right)+\sqrt{\mathbf{Z}[p,k]} \right) \left\lVert \mathbf{z}_{1} - \mathbf{z}_{2} \right\rVert_{2} \nonumber\\
+&\left\lVert \mathbf{z}_{2} \right\rVert_{2}\underbrace{\left\lvert \rho_{k,p}(\mathbf{z}_{1})- \rho_{k,p}(\mathbf{z}_{2}) \right\rvert}_{p_{1}},
\label{eq:ine2}
\end{align}
where the second inequality comes from the fact that $\varphi_{\mu}\left(\left\lvert\mathbf{f}_{k}^{H}\mathbf{g}_{p}(\mathbf{z}_{1}) \right \rvert\right)\geq \mu$. The term $p_{1}$ in \eqref{eq:ine2} can be upper bounded as

\footnotesize{\begin{align}
	&p_{1}\leq \left\lvert \mathbf{f}_{k}^{H}\mathbf{g}_{p}(\mathbf{z}_{1}) - \mathbf{f}_{k}^{H}\mathbf{g}_{p}(\mathbf{z}_{2}) \right\rvert \nonumber\\
	&+\sqrt{\mathbf{Z}[p,k]}\left\lvert\frac{\mathbf{f}_{k}^{H}\mathbf{g}_{p}(\mathbf{z}_{1})}{\varphi_{\mu}\left(\left\lvert\mathbf{f}_{k}^{H}\mathbf{g}_{p}(\mathbf{z}_{1}) \right \rvert\right)} - \frac{\mathbf{f}_{k}^{H}\mathbf{g}_{p}(\mathbf{z}_{2})}{\varphi_{\mu}\left(\left\lvert\mathbf{f}_{k}^{H}\mathbf{g}_{p}(\mathbf{z}_{2}) \right \rvert\right)}\right\rvert\nonumber\\
	&\leq \left\lvert \mathbf{f}_{k}^{H}\mathbf{g}_{p}(\mathbf{z}_{1}) - \mathbf{f}_{k}^{H}\mathbf{g}_{p}(\mathbf{z}_{2}) \right\rvert\nonumber\\
	&+\frac{\sqrt{\mathbf{Z}[p,k]}}{\mu^{2}}\varphi_{\mu}\left(\left\lvert\mathbf{f}_{k}^{H}\mathbf{g}_{p}(\mathbf{z}_{2}) \right \rvert\right)\left\lvert \mathbf{f}_{k}^{H}\mathbf{g}_{p}(\mathbf{z}_{1}) - \mathbf{f}_{k}^{H}\mathbf{g}_{p}(\mathbf{z}_{2}) \right\rvert \nonumber\\
	&+\frac{\sqrt{\mathbf{Z}[p,k]}}{\mu^{2}}\left\lvert \mathbf{f}_{k}^{H}\mathbf{g}_{p}(\mathbf{z}_{2})\right\rvert \left\lvert \varphi_{\mu}\left(\left\lvert\mathbf{f}_{k}^{H}\mathbf{g}_{p}(\mathbf{z}_{1}) \right \rvert\right) - \varphi_{\mu}\left(\left\lvert\mathbf{f}_{k}^{H}\mathbf{g}_{p}(\mathbf{z}_{2}) \right \rvert\right) \right\rvert.
	\label{eq:ine3}
	\end{align}}\normalsize

Recall that $\mathcal{J}$ is a closed bounded set, and thus compact. Since $\varphi_{\mu}(\cdot)$ is a continuous function, there exists a constant $M_{\varphi_{\mu}}$ such that $\varphi_{\mu}\left(\left\lvert\mathbf{f}_{k}^{H}\mathbf{g}_{p}(\mathbf{z}) \right \rvert\right)\leq M_{\varphi_{\mu}}$ for all $\mathbf{z}\in \mathcal{J}$. Also, from Lemma 2 in \cite{8410803} we have that $\varphi_{\mu}(\cdot)$ is a 1-Lipschitz function. Combining this with \eqref{eq:ine3} we get
\begin{align}
p_{1}&\leq \left\lvert \mathbf{f}_{k}^{H}\mathbf{g}_{p}(\mathbf{z}_{1}) - \mathbf{f}_{k}^{H}\mathbf{g}_{p}(\mathbf{z}_{2}) \right\rvert \nonumber\\
+&\frac{\sqrt{\mathbf{Z}[p,k]}M_{\varphi_{\mu}}}{\mu^{2}}\left\lvert \mathbf{f}_{k}^{H}\mathbf{g}_{p}(\mathbf{z}_{1}) - \mathbf{f}_{k}^{H}\mathbf{g}_{p}(\mathbf{z}_{2}) \right\rvert \nonumber\\
+&\frac{\sqrt{\mathbf{Z}[p,k]}}{\mu^{2}}\left\lvert \mathbf{f}_{k}^{H}\mathbf{g}_{p}(\mathbf{z}_{2})\right\rvert \left\lvert \left\lvert\mathbf{f}_{k}^{H}\mathbf{g}_{p}(\mathbf{z}_{1}) \right \rvert - \left\lvert\mathbf{f}_{k}^{H}\mathbf{g}_{p}(\mathbf{z}_{2}) \right \rvert \right\rvert,
\end{align}
and thus
\begin{align}
&p_{1}\leq \left(\frac{\sqrt{\mathbf{Z}[p,k]}M_{\varphi_{\mu}}}{\mu^{2}}+ 1\right)\left\lvert \mathbf{f}_{k}^{H}\mathbf{g}_{p}(\mathbf{z}_{1}) - \mathbf{f}_{k}^{H}\mathbf{g}_{p}(\mathbf{z}_{2}) \right\rvert \nonumber\\
+&\frac{\sqrt{\mathbf{Z}[p,k]}}{\mu^{2}}\left\lvert \mathbf{f}_{k}^{H}\mathbf{g}_{p}(\mathbf{z}_{2})\right\rvert  \left\lvert\mathbf{f}_{k}^{H}\mathbf{g}_{p}(\mathbf{z}_{1})  - \mathbf{f}_{k}^{H}\mathbf{g}_{p}(\mathbf{z}_{2})  \right\rvert,
\label{eq:ine4}
\end{align}
where \eqref{eq:ine4} results from applying the triangular inequality. Putting together \eqref{eq:ine2} and \eqref{eq:ine4} we obtain that
\begin{align}
&\frac{1}{\sqrt{N}}\left\lVert f_{k,p}(\mathbf{z}_{1}) -  f_{k,p}(\mathbf{z}_{2}) \right\rVert_{2} \nonumber\\
\leq& \frac{\left\lvert\mathbf{f}_{k}^{H}\mathbf{g}_{p}(\mathbf{z}_{1})\right\rvert}{\mu}\left(M_{\varphi_{\mu}} +\sqrt{\mathbf{Z}[p,k]}\right) \left\lVert \mathbf{z}_{1} - \mathbf{z}_{2} \right\rVert_{2}\nonumber\\
+&\left\lVert \mathbf{z}_{2} \right\rVert_{2}\left(\frac{\sqrt{\mathbf{Z}[p,k]}M_{\varphi_{\mu}}}{\mu^{2}}+ 1\right)\left\lvert \mathbf{f}_{k}^{H}\mathbf{g}_{p}(\mathbf{z}_{1}) - \mathbf{f}_{k}^{H}\mathbf{g}_{p}(\mathbf{z}_{2}) \right\rvert \nonumber\\
+&\frac{\left\lVert \mathbf{z}_{2} \right\rVert_{2}\sqrt{\mathbf{Z}[p,k]}}{\mu^{2}}\left\lvert \mathbf{f}_{k}^{H}\mathbf{g}_{p}(\mathbf{z}_{2})\right\rvert \left\lvert\mathbf{f}_{k}^{H}\mathbf{g}_{p}(\mathbf{z}_{1})  - \mathbf{f}_{k}^{H}\mathbf{g}_{p}(\mathbf{z}_{2})  \right\rvert.
\label{eq:ine5}
\end{align}

Observe that the upper bound in \eqref{eq:ine5} directly depends on a term of the form $\mathbf{f}_{k}^{H}\mathbf{g}_{p}(\mathbf{z})$ for some $\mathbf{z}\in \mathcal{J}$, which might be zero. However, Lemma \ref{lem:generic} proves that $\left\lvert\mathbf{f}_{k}^{H}\mathbf{g}_{p}(\mathbf{z})\right\rvert>0$ or equivalently $\mathbf{f}_{k}^{H}\mathbf{g}_{p}(\mathbf{z})\not = 0$, for almost all $\mathbf{z}\in \mathcal{J}$. 
\begin{lemma}
	Let $\mathbf{z}\in \mathcal{J}$ where $\mathcal{J}$ as defined in \eqref{eq:covergenceset}. Then, for almost all $\mathbf{z}\in \mathcal{J}$ the following holds
	\begin{align}
	\left\lvert\mathbf{f}_{k}^{H}\mathbf{g}_{p}(\mathbf{z})\right\rvert>0,
	\end{align}
	for all $k,p\in \{0,\cdots,N-1\}$, with $\mathbf{g}_{p}(\mathbf{z})$ as in \eqref{eq:auxvector}.
	\label{lem:generic}
\end{lemma}
\begin{proof}
	We prove this lemma by contradiction. Suppose that $\left\lvert\mathbf{f}_{k}^{H}\mathbf{g}_{p}(\mathbf{z})\right\rvert=0$. Then, from \eqref{eq:system} we have that
	\begin{align}
	&\left\lvert\mathbf{f}_{k}^{H}\mathbf{g}_{p}(\mathbf{z})\right\rvert^{2} = \left\lvert\sum_{n=0}^{N-1} \mathbf{z}[n]\mathbf{z}[n+pL]e^{-2\pi ink/N} \right \rvert^{2}\nonumber\\
	&=\sum_{n,m=0}^{N-1}\left(\mathbf{z}[n]\overline{\mathbf{z}}[m]\mathbf{z}[n+pL]\overline{\mathbf{z}}[m+pL]\right)e^{\frac{2\pi i(m-n)k}{N}}=0.
	\label{eq:quadraticpolynomical}
	\end{align}
	Observe that \eqref{eq:quadraticpolynomical} is a quartic polynomial equation with respect to the entries of $\mathbf{z}$. However, for almost all signals $\mathbf{z}\in \mathcal{J}$ the left hand side of \eqref{eq:quadraticpolynomical} will not be equal to zero which leads to a contradiction \cite{bendory2017signal}. 
\end{proof}

Then, proceeding to bound the term $\left\lvert\mathbf{f}_{k}^{H}\mathbf{g}_{p}(\mathbf{z})\right\rvert$, notice that from \eqref{eq:system} we have that
\begin{align}
\left\lvert\mathbf{f}_{k}^{H}\mathbf{g}_{p}(\mathbf{z})\right\rvert &= \left\lvert\sum_{n=0}^{N-1} \mathbf{z}[n]\mathbf{z}[n+pL]e^{-2\pi ink/N} \right \rvert\nonumber\\
&\leq \sum_{n=0}^{N-1} \left\lvert\mathbf{z}[n]\mathbf{z}[n+pL]\right \rvert \leq N\lVert \mathbf{z} \rVert_{2},
\label{eq:ine6}
\end{align}
in which the second inequality arises from $\lVert \mathbf{z} \rVert_{2} \leq \sqrt{N}\lVert \mathbf{z} \rVert_{\infty}$ and $\lVert \mathbf{z} \rVert_{1}\leq \sqrt{N}\lVert \mathbf{z} \rVert_{2}$.	Combining \eqref{eq:ine5} and \eqref{eq:ine6} we get
\begin{align}
\frac{1}{\sqrt{N}}&\left\lVert f_{k,p}(\mathbf{z}_{1}) -  f_{k,p}(\mathbf{z}_{2}) \right\rVert_{2}\nonumber\\
\leq& \frac{N\lVert \mathbf{z}_{1} \rVert_{2}}{\mu}\left(M_{\varphi_{\mu}} +\sqrt{\mathbf{Z}[p,k]}\right) \left\lVert \mathbf{z}_{1} - \mathbf{z}_{2} \right\rVert_{2}\nonumber\\
+&\left\lVert \mathbf{z}_{2} \right\rVert_{2}\left(\frac{\sqrt{\mathbf{Z}[p,k]}M_{\varphi_{\mu}}}{\mu^{2}}+ 1\right)\left\lvert \mathbf{f}_{k}^{H}\mathbf{g}_{p}(\mathbf{z}_{1}) - \mathbf{f}_{k}^{H}\mathbf{g}_{p}(\mathbf{z}_{2}) \right\rvert \nonumber\\
+&\frac{N\lVert \mathbf{z}_{2}\rVert^{2}_{2}\sqrt{\mathbf{Z}[p,k]}}{\mu^{2}}  \left\lvert\mathbf{f}_{k}^{H}\mathbf{g}_{p}(\mathbf{z}_{1})  - \mathbf{f}_{k}^{H}\mathbf{g}_{p}(\mathbf{z}_{2})  \right\rvert.
\label{eq:ine7}
\end{align}

Now, we have to analyze the term $\left\lvert\mathbf{f}_{k}^{H}\mathbf{g}_{p}(\mathbf{z}_{1})  - \mathbf{f}_{k}^{H}\mathbf{g}_{p}(\mathbf{z}_{2})  \right\rvert$ in \eqref{eq:ine7}. Specifically, from \eqref{eq:system} it can be obtained that
\begin{align}
&\left\lvert\mathbf{f}_{k}^{H}\mathbf{g}_{p}(\mathbf{z}_{1})  - \mathbf{f}_{k}^{H}\mathbf{g}_{p}(\mathbf{z}_{2}) \right\rvert \nonumber\\
&\leq \sum_{n=0}^{N-1} \left\lvert\mathbf{z}_{1}[n]\mathbf{z}_{1}[n+pL] - \mathbf{z}_{2}[n]\mathbf{z}_{2}[n+pL]\right \rvert\nonumber\\
&\leq N\left(\lVert \mathbf{z}_{1}\rVert_{2} + \lVert \mathbf{z}_{2}\rVert_{2}\right) \lVert \mathbf{z}_{1} - \mathbf{z}_{2}\rVert_{2},
\label{eq:ine8}
\end{align}
where the second inequality results from $\lVert \mathbf{z} \rVert_{2} \leq \sqrt{N}\lVert \mathbf{z} \rVert_{\infty}$ and $\lVert \mathbf{z} \rVert_{1}\leq \sqrt{N}\lVert \mathbf{z} \rVert_{2}$. Combining \eqref{eq:ine7} and \eqref{eq:ine8} we obtain that
\begin{align}
\left\lVert f_{k,p}(\mathbf{z}_{1}) -  f_{k,p}(\mathbf{z}_{2}) \right\rVert_{2} \leq r_{k,p} \lVert \mathbf{z}_{1} - \mathbf{z}_{2} \rVert_{2},
\label{eq:secondresult}
\end{align}
where $r_{k,p} $ is given by
\begin{align}
&r_{k,p} =  \frac{N\sqrt{N}\lVert \mathbf{z}_{1} \rVert_{2}}{\mu}\left(M_{\varphi_{\mu}} +\sqrt{\mathbf{Z}[p,k]}\right) \nonumber\\
&+N\sqrt{N}\left(\lVert \mathbf{z}_{1}\rVert_{2} + \lVert \mathbf{z}_{2}\rVert_{2}\right)\lVert \mathbf{z}_{2}\rVert_{2}\left(\frac{\sqrt{\mathbf{Z}[p,k]}M_{\varphi_{\mu}}}{\mu^{2}}+ 1\right) \nonumber\\
&+N^{2}\sqrt{N}\left(\lVert \mathbf{z}_{1}\rVert_{2} + \lVert \mathbf{z}_{2}\rVert_{2}\right)\frac{\lVert \mathbf{z}_{2}\rVert^{2}_{2}\sqrt{\mathbf{Z}[p,k]}}{\mu^{2}}.
\label{eq:ine9}
\end{align}
Since the set $\mathcal{J}$ is bounded, then $\lVert \mathbf{z} \rVert_{2}<\infty$ for all $\mathbf{z}\in \mathcal{J}$. Therefore, $0<r_{k,p}<\infty$, and from \eqref{eq:secondresult} the result holds.\vspace{1em}

4) We proceed to prove \eqref{eq:variance}. Observe that
\begin{align}
&\mathbb{E}_{\Gamma_{(t)}}\left[\left\lVert \mathbf{d}_{\Gamma_{(t)}}  - \frac{\partial h(\mathbf{z},\mu)}{\partial \overline{\mathbf{z}}} \right\rVert^{2}_{2}\right]\nonumber\\
\leq& \mathbb{E}_{\Gamma_{(t)}}\left[ 2\left\lVert \mathbf{d}_{\Gamma_{(t)}}\right\rVert^{2}_{2}\right]  + 2\left\lVert\frac{\partial h(\mathbf{z},\mu)}{\partial \overline{\mathbf{z}}}\right\rVert^{2}_{2},
\label{eq:final1}
\end{align}
in which the inequality comes from the fact that $\left\lVert \mathbf{w}_{1} + \mathbf{w}_{2} \right\rVert^{2}_{2}\leq 2 \left(\left\lVert \mathbf{w}_{1} \right\rVert^{2}_{2}+\left\lVert \mathbf{w}_{2} \right\rVert^{2}_{2}\right)$ for any $\hspace{5em}\mathbf{w}_{1},\mathbf{w}_{2}\in \mathbb{C}^{N}$. Combining \eqref{eq:assumption} and \eqref{eq:final1} we have that
\begin{align}
\mathbb{E}_{\Gamma_{(t)}}\left[\left\lVert \mathbf{d}_{\Gamma_{(t)}}  - \frac{\partial h(\mathbf{z},\mu)}{\partial \overline{\mathbf{z}}} \right\rVert^{2}_{2}\right] &\leq  \mathbb{E}_{\Gamma_{(t)}}\left[ 2\left\lVert \mathbf{d}_{\Gamma_{(t)}}\right\rVert^{2}_{2}\right]  + 2U\left\lVert\mathbf{z}\right\rVert^{2}_{2},
\label{eq:final3}
\end{align}
for some $U>0$. Recall that $\Gamma_{(t)}$ is sampled uniformly at random from all subsets of $\{1,\cdots,N \}\times \{1\cdots,R\}$ with cardinality $Q$. From the definition of $\mathbf{d}_{\Gamma_{(t)}}$ in Line 9 of Algorithm \ref{alg:smothing}, it can be concluded that 
\begin{align}
\mathbb{E}_{\Gamma_{(t)}}\left[ 2\left\lVert \mathbf{d}_{\Gamma_{(t)}}\right\rVert^{2}_{2}\right] &\leq \frac{4Q}{N^{2}}\sum_{p,k=0}^{N-1} \left\lVert f_{k,p}(\mathbf{z}) + g_{k,p}(\mathbf{z})\right\rVert^{2}_{2}\nonumber\\
&\leq \frac{8Q}{N^{2}}\sum_{p,k=0}^{N-1} \left\lVert f_{k,p}(\mathbf{z})\right\rVert^{2}_{2} +  \left\lVert g_{k,p}(\mathbf{z})\right\rVert^{2}_{2},
\end{align}
using the fact that $\left\lVert \mathbf{w}_{1} + \mathbf{w}_{2} \right\rVert^{2}_{2}\leq 2 \left(\left\lVert \mathbf{w}_{1} \right\rVert^{2}_{2}+\left\lVert \mathbf{w}_{2} \right\rVert^{2}_{2}\right)$ for any $\mathbf{w}_{1},\mathbf{w}_{2}\in \mathbb{C}^{N}$. Furthermore, since $f_{k,p}(\mathbf{z})$ and $g_{k,p}(\mathbf{z})$ satisfy \eqref{eq:functionf} and \eqref{eq:functiong}, respectively, we conclude that
\begin{align}
\mathbb{E}_{\Gamma_{(t)}}\left[ 2\left\lVert \mathbf{d}_{\Gamma_{(t)}}\right\rVert^{2}_{2}\right] \leq \frac{8Q\left\lVert\mathbf{z}\right\rVert^{2}_{2}}{N^{2}}\sum_{p,k=0}^{N-1} r^{2}_{k,p} + s^{2}_{k,p},
\label{eq:final2}
\end{align}
for some constants $r_{k,p},s_{k,p}>0$. Thus, combining \eqref{eq:final3} and \eqref{eq:final2} we obtain that
\begin{align}
\mathbb{E}_{\Gamma_{(t)}}\left[\left\lVert \mathbf{d}_{\Gamma_{(t)}}  - \frac{\partial h(\mathbf{z},\mu)}{\partial \overline{\mathbf{z}}} \right\rVert^{2}_{2}\right] \leq \zeta^{2},
\label{eq:final4}
\end{align}
where $\zeta$ is defined as
\begin{align}
\zeta = \left\lVert\mathbf{z}\right\rVert_{2}\sqrt{\frac{8Q}{N^{2}}\sum_{p,k=0}^{N-1} r^{2}_{k,p} + s^{2}_{k,p} + 2U}.
\end{align}
Notice $\zeta<\infty$ because the set $\mathcal{J}$ is bounded. Thus, from \eqref{eq:final4} the result holds.

\bibliography{sample}
\bibliographystyle{ieeetr}

\end{document}